%% file: main.tex
\documentclass{sig-alternate}

\usepackage{array}
\usepackage[figuresright]{rotating}
\usepackage{multirow}
\usepackage[table]{xcolor}
\usepackage{footnote}
\usepackage{listings}
\usepackage{color}
\usepackage{url}
\usepackage{comment}
\usepackage{algorithm, algpseudocode}
\usepackage{subcaption}
\usepackage{amsmath, amssymb, mathptmx}
\usepackage{placeins, float}
\usepackage{pgfplots}
\usetikzlibrary{patterns}
\usepgfplotslibrary{external}

\definecolor{dkgreen}{rgb}{0,0.6,0}
\definecolor{gray}{rgb}{0.5,0.5,0.5}
\definecolor{mauve}{rgb}{0.58,0,0.82}

\newcommand{\mkcomment}[1]{\noindent {\bf \\MK : #1 \\}}
\newcommand{\mmf}{{\sc MMF}}
\newcommand{\mmfsl}{{\sc MMFsl}}
\newcommand{\mmfsf}{{\sc MMFsf}}
\newcommand{\pfsl}{{\sc FastPFsl}}
\newcommand{\pfsf}{{\sc FastPFsf}}
\newcommand{\fastpf}{{\sc FastPF}}
\newcommand{\static}{{\sc Static}}
\newcommand{\opt}{{\sc OptP}}
\newcommand{\robus}{{\sf ROBUS}}

\newtheorem{theorem}{Theorem}

\newtheorem{newdefn}{Definition}

\newtheorem{lemma}{Lemma}
\newcommand{\squishlist}{
 \begin{list}{$\bullet$}
  { \setlength{\itemsep}{0pt}
     \setlength{\parsep}{3pt}
     \setlength{\topsep}{3pt}
     \setlength{\partopsep}{0pt}
     \setlength{\leftmargin}{1.5em}
     \setlength{\labelwidth}{1em}
     \setlength{\labelsep}{0.5em} } }

\newcommand{\squishlisttwo}{
 \begin{list}{}
  { \setlength{\itemsep}{0pt}
    \setlength{\parsep}{0pt}
    \setlength{\topsep}{0pt}
    \setlength{\partopsep}{0pt}
    \setlength{\leftmargin}{1.5em}
    \setlength{\labelwidth}{1.5em}
    \setlength{\labelsep}{0.5em} } }

\newcommand{\squishend}{
  \end{list}  }

\begin{document}






\title{ROBUS: Fair Cache Allocation for Multi-tenant Data-parallel Workloads}

\numberofauthors{1}
\author{
\alignauthor
Mayuresh Kunjir, Brandon Fain, Kamesh Munagala, Shivnath Babu\\
	\affaddr{Duke University}\\
	\email{\{mayuresh, btfain, kamesh, shivnath\}@cs.duke.edu}
}


\maketitle

\input{abstract}

\input{introduction}

\input{system}

\input{model}

\input{algos}

\input{evaluation}

\input{related}

\input{conclusion}


\bibliographystyle{abbrv}
\bibliography{refs}  

\pagebreak

\appendix

\input{appendix}

\end{document}

%% file: abstract.tex
\begin{abstract}
Systems for processing big data---e.g., Hadoop, Spark, and massively parallel databases---need to run workloads on behalf of multiple tenants simultaneously. The abundant disk-based storage in these systems is usually complemented by a smaller, but much faster, {\em cache}. Cache is a precious resource: Tenants who get to use cache can see two orders of magnitude performance improvement. Cache is also a limited and hence shared resource: Unlike a resource like a CPU core which can be used by only one tenant at a time, a cached data item can be accessed by multiple tenants at the same time. Cache, therefore, has to be shared by a multi-tenancy-aware policy across tenants, each having a unique set of priorities and workload characteristics. 

In this paper, we develop cache allocation strategies that speed up the overall workload while being {\em fair} to each tenant. We build a novel fairness model targeted at the shared resource setting that incorporates not only the more standard concepts of Pareto-efficiency and sharing incentive, but also define envy freeness via the notion of {\em core} from cooperative game theory.  Our cache management platform, \robus, uses randomization over small time batches, and we develop a proportionally fair allocation mechanism that satisfies the core property in expectation. We show that this algorithm and related fair algorithms can be approximated to arbitrary precision in polynomial time. We evaluate these algorithms on a \robus\ prototype implemented on Spark with RDD store used as cache. Our evaluation on an industry-standard workload shows that our algorithms provide a speedup close to performance optimal algorithms while guaranteeing fairness across tenants.
\end{abstract}


%% file: introduction.tex
\section{Introduction} 
\label{sec:intro}

Two recent trends in data processing are: (i) 
the use of multi-tenant clusters
for analyzing large and diverse datasets, and (ii)
the aggressive use of memory to speed up 
processing by caching datasets. The 
growing popularity of systems like Apache Spark~\cite{spark}, 
SAP HANA~\cite{hana},
Hadoop with Discardable Distributed Memory~\cite{ddm}, and
Tachyon~\cite{tachyon} highlight these trends. 

For example, Spark introduces an abstraction called
{\em Resilient Distributed Dataset (RDD)} to
represent any data relevant to modern analytics: 
files (on a local or distributed
file-system), tables (horizontally
or vertically partitioned), vertices or edges of graphs,
statistical models learned from data, etc.
 A user can create an RDD
directly from data residing on a local or distributed file-system,
or by applying a transformation to one or more other RDDs. 
The user can then direct the system to cache the
RDD in memory. Figure \ref{fig:spark-program} gives an example. 
Computations done on RDDs cached in memory
run 10-100x faster than when the data resides on disk~\cite{spark}.

\begin{figure}[h]
\centering
\includegraphics[width=0.9\columnwidth]{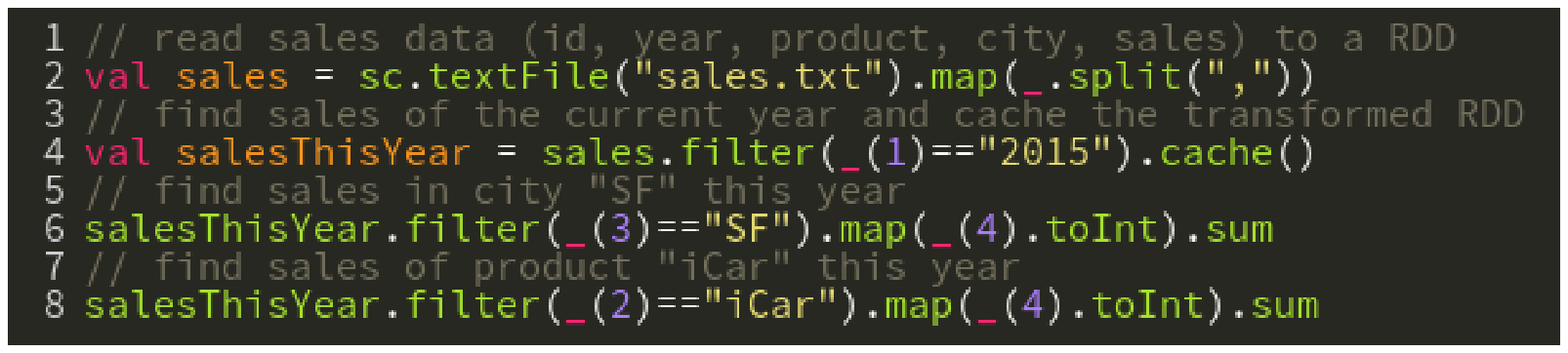}
\caption{A sample Spark program}
\label{fig:spark-program}
\end{figure}
\vspace*{-0.4cm}

{\bf Can User-Directed Caching and Multi-tenancy Coexist?}
User-directed caching brings some major challenges in a multi-tenant 
data analytics cluster:

\squishlist

\item {\em A precious resource:} 
Tenants who get to use the in-memory cache see many orders of magnitude
of performance improvement. However, the cache is also a limited
resource since the total size of memory in a cluster is
usually orders of magnitude smaller than the
data sizes stored and queried in the cluster.

\item {\em Complications from sharing:} Unlike a resource like a CPU core
which is used by one tenant at a time, a cached data item can 
simultaneously benefit a high-priority and a low-priority tenant.

\item {\em Avoiding cache hogs:} 
Low-priority tenants should not be able to hog the available cache,
and prevent other tenants from getting the performance benefits
they deserve.

\item {\em Utilities differ:} Different tenants have different
utilities for datasets that could be placed in the cache.

\squishend

\vspace{1mm}
\noindent When faced with such challenges, traditional cache
allocation policies can lead to
user dissatisfaction, poor or unpredictable
performance, and low resource utilization.
We will illustrate
the problems and opportunities through an example.
Consider a social-networking company, SpaceBook,
that runs a multi-tenant cluster for analyzing
datasets about how its users are using the service.

\noindent {\bf Multiple tenants:} 
The predominant practice in the industry is to group
similar users---e.g., users in the same department---into
queues (or, pools). Each queue forms a tenant in the cluster.
The cluster at SpaceBook is used by three tenants:
(i) {\em Analyst}, the business analysts in the company,
(ii) {\em Engineer}, the developers in the company who
develop data-driven applications such as recommendation models,
and (iii) {\em VP}, the top-level management in the company
such as the CEO and the Chief Security Officer
who look at hourly and daily reports.

\noindent {\bf Cacheable entities:}
These three tenants will benefit from caching one or more 
of three {\em views}---{\em R}, {\em S}, and {\em P}---each 
of size {\em M} bytes. 
Throughout this paper, ``view''  refers to any data item that can be cached to 
give a performance benefit. For SQL workloads, a view
corresponds to a SQL expression, like any candidate
view generated by a materialized view selection algorithm 
~\cite{gupta,yang,automated,dynamat}.
For broader data analytics---e.g., 
machine learning and graph processing---a view corresponds 
to a dataset on which the user has put a cache directive
(recall the example Spark program from Figure \ref{fig:spark-program}).

\begin{table}[h!]
\centering
{\small
 \begin{tabular}{|c ||c c c|} 
 \hline
 Tenant & $R$ & $S$ & $P$ \\ [0.5ex] 
 \hline\hline
 Analyst & 2 & 1 & 0 \\ 
 \hline
 Engineer & 2 & 1 & 0 \\
 \hline
 VP & 0 & 1 & 2 \\
 \hline
\end{tabular}
}
\caption{Utilities of cached views to tenants at SpaceBook}
\vspace*{-0.4cm}
\label{table-1-ref}
\end{table}

\noindent {\bf Utilities:}
The matrix in Table~\ref{table-1-ref} shows the {\em utility} that
each tenant gets if the corresponding view were to be cached in memory.
A simple definition of utility we will use in this paper
is the savings in I/O because data is
read from the in-memory cache versus disk. For example, if
view {\em R} is cached in memory, then tenant {\em Analyst} will get
a utility of two units. One common  pattern 
in multi-tenant clusters that we bring out in 
Table~\ref{table-1-ref} is that view {\em R} could be 
the detailed logs that business analysts and developers
access quite often; view {\em P} could be a table that only the top-level
management has access to; while view {\em S} could be a 
materialized view with aggregated information shared by all tenants.

\vspace{1mm}
\noindent {\bf Scenario 1:} Debbie, the cluster DBA, is 
responsible for allocating resources so that the tenants 
get good performance while all resources are used effectively. 
Suppose the in-memory cache at SpaceBook has a total size of {\em M} bytes.
Debbie first configures a static and equal partitioning of the cache so that 
each tenant is entitled to $\frac{M}{3}$ bytes of
cache memory. Recall that each of the views {\em R}, {\em S},
and {\em P} are {\em M} bytes each; so none of them
will fit in their $\frac{M}{3}$ bytes of cache.
Hence, resource utilization will be poor and none of the tenants
will receive any performance boost.

\vspace{1mm}
\noindent {\bf Scenario 2:}
Next, Debbie switches to a more common
cache allocation policy, {\em Least Recently
Used (LRU)}. The view {\em R} is used the most at Spacebook,
so it will likely remain cached for the most time. Thus, the Analyst
and Engineer tenants will see performance speedups. However,
the VP tenant's workload will see poor performance, causing these 
users---including Zuck, SpaceBook's CEO---to complain that important
reports needed for their decision-making are not being generated on time.

\vspace{1mm}
\noindent {\bf Scenario 3:}
Debbie decides to give the VP tenant 50\% higher priority
than the other tenants. So, she assigns
weights to the Analyst, Engineer, and VP tenants in the ratio
$1:1:1.5$. Debbie now switches to a policy that
allocates the cache based on the weighted utility of the tenants.
She tells Zuck that his reports will now be generated faster.
Unfortunately, Zuck will not see any
performance improvement: view {\em R} will
still be the only one cached since it has the highest weighted
utility of $4$ ($= 2 \times 1 + 2 \times 1$); higher than
view {\em S}'s weighted utility of $3.5$ ($= 1 \times 1 + 1
\times 1 + 1 \times 1.5$), and view {\em P}'s weighted utility
of $3$ ($= 2 \times 1.5$).

\vspace{1mm}
\noindent {\bf Scenario 4:}
To improve the poor performance seen by the VP tenant,
Zuck gives Debbie the money needed to double
the cache memory in the cluster.
Now two views will fit in the $2M$-sized cache.
However, even after this massive investment,
the VP tenant will only see a minor increase in performance
compared to the Analyst and Engineer tenants:
views {\em R} and {\em S} will now
be cached since they together have the highest weighted
utility of $7.5$ ($4$ for {\em R} + $3.5$ for {\em S}); higher than
$7$ for {\em R} and {\em P}, and $6.5$ for {\em S} and {\em P}.

\vspace{1mm}
\noindent {\bf Scenario 5:}
Zuck is very unhappy with Debbie. She now tries to
improve things by switching back to static partitioning
of $\frac{2M}{3}$ for every tenant, causing everyone 
to get poor performance because none of the views fit. In desperation, 
Debbie now has to go to the Analyst and Engineer tenants to
request them to stop adding cache directives to their workloads. 
The whole multi-tenant situation becomes extremely messy.

\vspace{1mm}
\noindent {\bf Better scenarios:}
Let us consider what Debbie would have wanted ideally.
An alternative in Scenario 3 is to 
cache view {\em S} instead of {\em R}. While {\em S}
has a slightly lower weighted utility of $3.5$ compared
to $4$ for $R$, all three tenants will see peformance improvements
from caching {\em S}. An alternative in Scenario 4 is to
cache {\em R} and {\em P} which will also 
give performance benefits to all three
tenants while only being slightly
lower in overall weighted utility than 
caching {\em R} and {\em S}. In particular, the {\em VP} tenant will 
now see major benefits from doubling the cache size.

The above example shows non-trivial nature of the problem of cache allocation in multi-tenant setups. There is a need to make principled choices when it comes to picking data items to cache. This motivates the main challenge we address.

\begin{quote}
\em Develop a cache allocation policy that provides near-optimal performance speedups for tenants' workload while simultaneously achieving near-optimal fairness in terms of the tenants' performance.
\end{quote}

\subsection*{Our Contributions.} 
\begin{itemize}
 \item In Section~\ref{sec:archi}, we propose \robus, a platform to optimize multi-tenant query workloads in an {\em online} manner using cache for speedup. This framework groups queries in small time-based batches and employs a randomized cache allocation policy on each batch.
 
 \item In Section~\ref{sec:model}, we consider the abstract setting of
   shared resource allocation within a batch, and enumerate properties
   that we desire from any allocation scheme. We show that the notion of {\em core} from
   cooperative game theory captures the fairness properties
   in a succinct fashion. We show that when restricted to randomized
   allocation policies within a batch, a simple
   algorithm termed {\em proportional fairness} generates an
   allocation which satisfies fairness properties in expectation
   for that batch.
 
 \item The policies we construct are based on convex programming formulations of exponential size. Nevertheless, in Section~\ref{sec:algo}, we show that these policies admit to arbitrarily good approximations in polynomial time using the multiplicative weight update method. We present implementations of two fair policies: max-min fairness and proportional fairness. We also present faster and more practical heuristics for computing these solutions.
 
 \item We show a proof-of-concept implementation of \robus\ on a multi-tenant Spark cluster. Motivated by practical use cases, we develop a synthetic workload generator to create various scenarios. Implementation details and evaluation are provided in Section~\ref{sec:eval}. Results show that our policies provide desirable throughput and fairness in a comprehensive set of setups.

\item Finally, our policies are specified abstractly, and as such easily extend to other resource allocation settings. We discuss this in Section~\ref{sec:discuss-fair}. 
\end{itemize}

%% file: system.tex
\section{ROBUS Platform} 
\label{sec:archi}
\robus\ (Random Optimized Batch Utility Sharing), shown 
in Figure~\ref{fig:architecture}, 
is the  cache management platform we have developed 
for multi-tenant data-parallel workloads. 
\robus\ is designed to be easily pluggable in systems like
Hadoop and Spark. Each tenant submits
its workload in an online fashion 
to a designated queue which is characterized by a weight indicating the
tenant's {\em fair} share of system resources. (Recall
our example from Section \ref{sec:intro}.)

\begin{figure}[h]
\centering
\includegraphics[width=0.8\columnwidth]{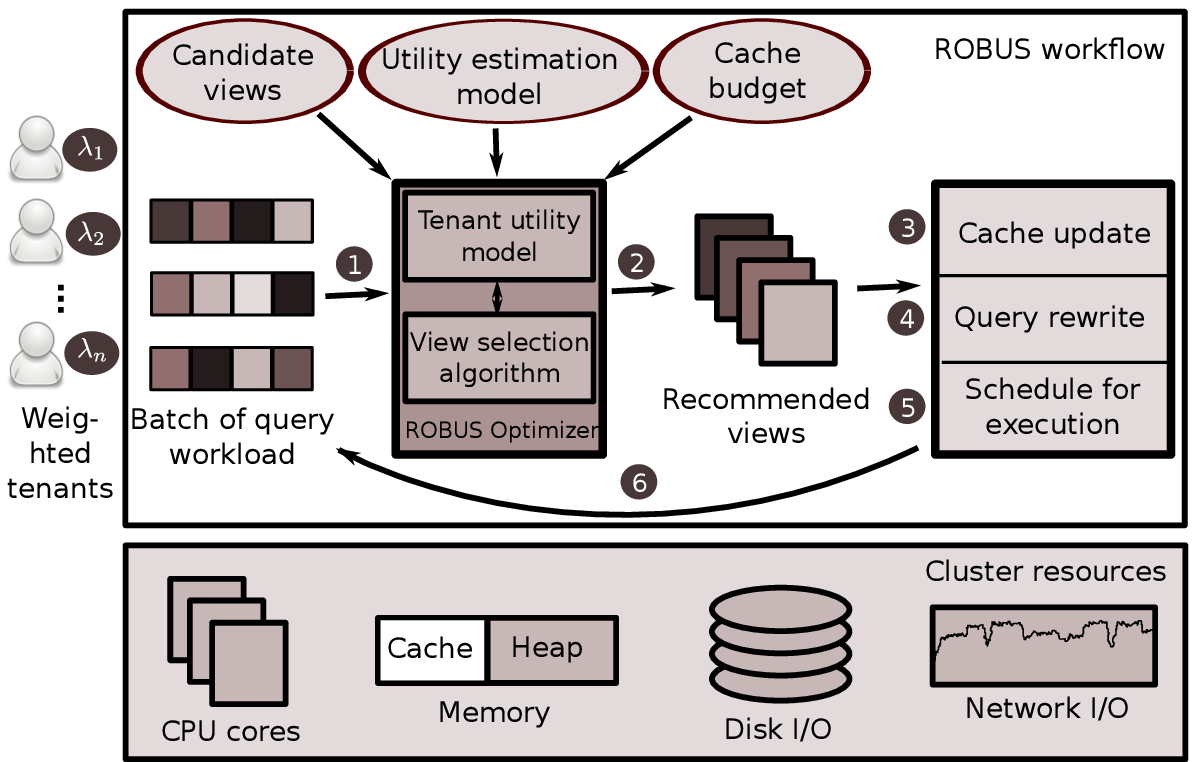}
\caption{\robus\ platform}
\label{fig:architecture}
\end{figure}


As illustrated in Figure~\ref{fig:architecture}, 
\robus\ processes the workload in {\em batches} 
by running five steps in a repeated loop. 
Step 1 removes a batch of queries 
that were submitted in a fixed time interval
into the tenants' queues. 
Step 2 runs an algorithm over this entire batch to select a set of
views to cache. This computation simultaneously optimizes performance and 
fairness; designing this algorithm is the main focus of this paper. 

Step 2 takes three inputs: (i) a set of candidate views
for the query batch, (ii) a utility estimation model for cached views, 
and (iii) the total cache budget (i.e., memory available for caching). 
The candidate view generation in \robus\ is a pluggable module. 
By default, the candidate views for a SQL query are the 
base tables accessed by the query. 
For workloads like machine learning and graph processing, 
the candidate views are datasets on which the user has put 
a cache directive (recall the example Spark program from 
Figure \ref{fig:spark-program}). 

Any candidate view selection algorithm
from the literature can be plugged in to 
\robus~\cite{gupta,yang,automated,dynamat}. To  support 
this feature, \robus\ has a pluggable Step 4 where a 
query can be rewritten to use the views selected for caching 
in Step 2 before being run in Step 5. We make use of 
\robus's pluggability in Section~\ref{sec:eval} to run 
a candidate view selection algorithm that considers 
different vertical projections of input tables in the workload.
In future, we plan to extend Step 4 to support re-optimization
of the query based on the cached views. Re-optimization
may change the query plan entirely. 

The utility estimation model is used in the view selection procedure
to estimate the utility provided to a query by any cached view. 
\robus\ currently models these utilities as savings in disk I/O costs
if the view were to be read off of in-memory cache versus disk. 
This approach keeps the models simple and widely applicable. In future,
we plan to incorporate richer utility models from the literature 
that can account for more complex candidate view selection algorithms
that consider interactions among views (e.g., use of views 
can completely change the query plan for a query ~\cite{miso}). Total utility of a cache configuration to a tenant is computed by summing up estimated utilities of the queries submitted by the tenant. The tenant utilities thus computed are used by view selection algorithm to recommend optimal set of views to cache.

In Step 3, the cache is updated with the views 
selected by our algorithm (if they are not already in the cache).
The query batch is then run via Steps 4 and 5. 
Every query runs as data-parallel tasks
in a system like Hadoop or Spark. Our current prototype of \robus\
runs on Spark. 
A task scheduler (e.g.,~\cite{hadoop-fair, drf}) is responsible for 
allocating system resources to the tasks.
Cluster memory is divided into two parts: a heap space for 
run-time objects and a cache for 
the selected views. While the heap is divided across tasks and is allocated by the task scheduler, the cache is shared by all the queries 
in the batch simultaneously and managed by \robus.

%% file: model.tex
\section{Fairness Properties and Policies for Single Batch}
\label{sec:model}
In this section we study various notions of fairness when restricted to view selection for queries from a single batch. We consider policies that compute allocations that simultaneously provide large utility to many tenants, and enforce a rigorous notion of fairness between the tenants. Since this is very related to other resource allocation problems in economics~\cite{budish,abdulkadiroglu,BMoulin}, we draw heavily on that work for inspiration. However, the key difference from standard resource allocation problems is that in our setting, the resources (or views) are simultaneously shared by tenants. In contrast, the resource allocation settings in economics have typically considered {\em partitioning} resources between tenants. As we shall see below, this leads to interesting differences in the notions of fairness.

\subsection{Fairness and Randomization}
It is well-known in economics~\cite{random} that the combination of fairness and indivisible resources (in our case, the cache and views) necessitates randomization. To develop intuition, we present two examples. 

First consider a simple fair allocation scheme that for $N$ tenants simply allows each tenant to use $\frac{1}{N}$ of the total cache for her preferred view(s).  It is plausible that some tenants prefer a large view that does not fit in this partition but does fit in the cache.  Therefore, letting tenants have $\frac{1}{N}$ probability of using the whole cache can have arbitrarily larger expected utility than the scheme which with probability 1 lets them use $\frac{1}{N}$ fraction of the whole cache. 

Next,  consider a batch wherein two tenants each request a different large view such that only one can fit into the cache. In this case, there can be no deterministic allocation scheme that does not ignore one of the tenants.  Using randomization, we can easily ensure that each tenant has the same utility in expectation.  In fact, utility in expectation will be the per batch guarantee we seek, which over the long time horizon of a workload will lead to deterministic fairness.

\paragraph{Notation for Single Batch}
Since our view selection policy works on individual batches at a time, the notation and discussion below is specific to queries within a batch. Let $N$ denote the total number of tenants. Define:
\begin{newdefn}
A {\em configuration} $S$ is the set of views feasible in that the sum of the view sizes $\sum_{S_i \in S} |S_i|$ is at most the cache size.
\end{newdefn}
$U_i(S)$ denotes the utility to tenant $i$ that would result from caching $S$, which is defined as the sum over all queries in $i$'s queue of the utility for that query. 
 
\robus\ generates a set $Q$ of configurations which by definition can fit in the cache, and assigns a probability $x_S$ to cache each configuration $S \in Q$. Define the vector of all such probabilities as:
\begin{newdefn}
An {\em Allocation} $\mathbf{x}$ is the vector corresponding to probabilities $x_s$ of choosing configuration $S$ normalized so $\|\mathbf{x}\| = \sum_{S \in Q} x_S = 1$.
\end{newdefn}
We denote $U_i(\mathbf{x}) = \sum_{S \in Q} x_S U_i(S)$ as the expected utility of tenant $i$ in allocation $\mathbf{x}$. \robus\ implements allocation $\mathbf{x}$ by sampling a configuration from the probability distribution. 

For each tenant $i$, let $U^*_i = \max_{S} U_i(S)$ denote the maximum possible utility tenant $i$ can obtain if it were the only tenant in the system. For allocation $\mathbf{x}$, we define the {\em scaled utility} of $i$ as $V_i(\mathbf{x}) = \frac{U_i(\mathbf{x})}{U^*_i}$. We will use this concept crucially in defining our fairness notions.

\subsection{Basic Fairness Desiderata}
The first question to ask when designing a fair allocation algorithm is what properties define fairness.  There has been much recent work in economics and computer science on heterogeneous resource allocation problems and we begin by considering the properties that this related work examines \cite{drf,beyonddrf,noagentleftbehind}.  Note that because we work within a randomized model, all of these properties are framed in terms of expected utility of tenants.

\begin{itemize}
    \item \textbf{Pareto Efficiency (PE):} An allocation is Pareto-efficient if no other allocation simultaneously improves the expected utility of at least one tenant and does not decrease the expected utility of any tenant. 
    \item \textbf{Sharing Incentive (SI):} This property is termed {\em individual rationality} in Economics. For $N$ tenants, each tenant should expect higher utility in the shared allocation setting than she would expect from simply always having access to $\frac{1}{N}$ of the resources. Since our allocations are randomized, allocation $\mathbf{x}$ satisfies SI if for all allocations $\mathbf{y}$ with $||\mathbf{y}|| \le \frac{1}{N}$ and for tenants $i$, $U_i(\mathbf{x}) \ge U_i(\mathbf{y})$. In other words, $V_i(\mathbf{x}) \ge \frac{1}{N}$ for all tenants $i$, where $V_i(\mathbf{x})$ is the scaled utility function defined above.
\end{itemize}

One property that is widely studied in other resource allocation contexts is strategy-proofness on the part of the tenants (the notion that no tenant should benefit from lying). In our case, since the queries are seen by the query optimizer, strategy-proofness is not an issue.  The above desiderata also omit envy-freeness (that no tenant should prefer the allocation to another tenant) which is something we revisit later.

We now consider a progression of view selection mechanisms on a single batch from very simple to more sophisticated.  As a running example, suppose  there is a cache of capacity $1$. There are three views $R$, $S$, or $P$ that are demanded by $N$ tenants. Each view has unit size, so that we can cache only one view any time.  Note that this is a drastically simplified example setup only intended to build intuition about why certain view selection algorithms might fail or are superior to others; our results and experiments do not only have unit views, are not limited to three tenants, and may have arbitrarily complex utilities compared to these examples.

 We can summarize the input information our view selection might see in a given batch in a table (e.g., Table~\ref{table 1}) where the numbers represent utilities tenants get from the views. 
An allocation here is a vector $\mathbf{x}$ of three dimensions and $\|\mathbf{x}\|=1$ that gives the probabilities in our randomized framework $x_{R}, x_{S}, x_{P}$ for selecting the views.

\begin{table}[h!]
\centering
 \begin{tabular}{|c ||c c c|} 
 \hline
 Tenant & R & S & P \\ [0.5ex] 
 \hline\hline
 A & 1 & 0 & 0 \\ 
 \hline
 B & 0 & 1 & 0 \\
 \hline
 C & 0 & 0 & 1 \\ [1ex] 
 \hline
\end{tabular}
\caption{Every tenant gets utility from a different view}

\label{table 1}
\end{table}

\paragraph{Static Partitioning}
Recall that static partitioning is the algorithm that simply deterministically allows each of the $N$ tenants to use $\frac{1}{N}$ of the shared resource.  This algorithm does not take advantage of randomization.  For the example in Table~\ref{table 1}, this algorithm cannot cache anything because each user only gets to decide on the use of $\frac{1}{3}$ of the  cache.  The algorithm is sharing incentive in the standard deterministic setting, but is trivially not Pareto efficient and is not sharing incentive in expectation either. As mentioned previously, such examples motivate the randomization framework to start with.     

\paragraph{Random Serial Dictatorship} A natural progression from static partitioning is to consider random serial dictatorship (RSD), a mechanism that is widely considered~\cite{BMoulin,abdulkadiroglu} for problems such as house allocation and school choice. We order the tenants in a random permutation. Each tenant sequentially computes the best set of views to cache (in the residual cache space) to maximize its own utility.  In the example in Table~\ref{table 1}, each tenant gets a $\frac{1}{3}$ chance of picking her preferred resource (since in a random permutation each tenant has a $\frac{1}{3}$ chance of appearing first) so the allocation is $\mathbf{x}=<x_{R}=\frac{1}{3}, x_{S}=\frac{1}{3}, x_{P}=\frac{1}{3}>$, where each tenant has the same utility in expectation. In fact, it is easy to prove that RSD is always SI: Each tenant has $\frac{1}{N}$ chance of being first in the random ordering, so its scaled utility is at least $\frac{1}{N}$.

However, in contrast with resource partitioning problems, our problem has a shared aspect that RSD fails to capture.  For example, consider the situation in Table~\ref{table 2}; RSD computes the same allocation as in the example in Table: \ref{table 1}, $\mathbf{x}=<x_{R}=\frac{1}{3}, x_{S}=\frac{1}{3}, x_{P}=\frac{1}{3}>$. However, on this example, though RSD is SI, it is not Pareto-efficient (PE).  Tenants A and C have expected utility of 1 (a $\frac{1}{3}$ chance of getting $2$ if they come first in the permutation and a $\frac{1}{3}$ chance of getting $1$ if B does) and tenant B has expected utility of $\frac{1}{3}$ with this allocation.  However, if we used allocation $\mathbf{x}=<x_{R}=0, x_{S}=1, x_{P}=0>$ then tenants A, B, and C all have utility 1, which is strictly better for tenant B and as good for tenants A and C.  RSD fails to capture the fact that while each tenant may have different top preferences, many tenants may share secondary preferences. 

\begin{table}[h!]
\centering
 \begin{tabular}{|c ||c c c|} 
 \hline
 Tenant & R & S & P \\ [0.5ex] 
 \hline\hline
 A & 2 & 1 & 0 \\ 
 \hline
 B & 0 & 1 & 0 \\
 \hline
 C & 0 & 1 & 2 \\ [1ex] 
 \hline
\end{tabular}
\caption{Every tenant gets utility from the same view}

\label{table 2}
\end{table}

\paragraph{Utility Maximization Mechanism}
We next consider the mechanism which simply maximizes the total expected utility of an allocation, {\em i.e.}, $\mbox{arg\,max}_{\mathbf{x}} \sum_i U_i(\mathbf{x})$.  It is easy to check that this mechanism can ignore tenants who do not contribute enough to the overall utility.  In other words, it cannot be SI.  

\paragraph{Max-min Fairness (MMF)}
In this algorithm we combine previous insights to optimize performance subject to fairness constraints to get a mechanism that is both SI and PE.  For allocation $\mathbf{x}$, let $\mathbf{v}(\mathbf{x}) = (V_1(\mathbf{x}), V_2(\mathbf{x}), \ldots, V_N(\mathbf{x}))$ denote the vector of scaled utilities of the tenants. We choose an allocation $\mathbf{x}$ so that the vector $\mathbf{v}(\mathbf{x})$ is lexicographically max-min fair. This means the smallest value in $\mathbf{v}(\mathbf{x})$ is as large as possible; subject to this, the next smallest value is as large as possible, and so on. We present algorithms to compute these allocations in Section~\ref{sec:algo}.

\begin{theorem}
The MMF mechanism is both PE and SI.
\end{theorem}
\begin{proof}
The RSD mechanism guarantees scaled utility of at least $\frac{1}{N}$ to each tenant. Since the MMF allocation is lexicographically max-min, the minimum scaled utility it obtains is at least the minimum scaled utility in RSD, which is at least $\frac{1}{N}$. To show PE, note that if there were an allocation that yielded at least as large utility for all tenants, and strictly higher utility for one tenant, the new allocation would be lexicographically larger, contradicting the definition of MMF.
\end{proof}


\begin{table}[h!]
\centering
 \begin{tabular}{|c ||c c|} 
 \hline
 Tenant & R & S \\ [0.5ex] 
 \hline\hline
 $T_1$ & 1 & 0 \\ 
 \hline
 $T_2$ & 1 & 0 \\
 \hline
 ... & ... & ... \\  
 \hline
 $T_N$ & 0 & 1 \\ [1ex]
 \hline
\end{tabular}
\caption{All tenants except one get utility from the same view}

\label{table 4}
\end{table}

Consider the example in Table~\ref{table 4}. It is easy to see that the MMF value is $\frac{1}{2}$ and can be achieved with an allocation of $<x_{R}=\frac{1}{2}, x_{S}=\frac{1}{2}>$.  This allocation is both SI and PE.  

\subsection{Envy-freeness and the Core}
The above discussion omits one important facet of fairness. A fair allocation has to be {\em envy free}, meaning no tenant has to envy how the allocation treats another tenant. In the case where resources are partitioned between tenants, such a notion is easy to define: No tenant must derive higher utility from the allocation to another tenant. However, in our setting, resources (views) are shared between tenants, and the only common factor is the cache space. In any allocation $\mathbf{x}$, each tenant derives utility  from certain views, and we can term the expected size of these views as the {\em cache share} of this user.

One could try to define envy-freeness in terms of cache space as follows: No tenant should be able to improve expected utility by obtaining the cache share of another tenant. But this simply means all tenants have the same expected cache share. Such an allocation need not be Sharing Incentive. Consider the example in Table~\ref{table 5}, where each view $R$ and $S$ has size $1$ and the cache has size $1$. The only allocation that equalizes cache share caches $S$ entirely. But this is not SI for tenant $B$.

\begin{table}[h!]
\centering
 \begin{tabular}{|c ||c c|} 
 \hline
 Tenant & R & S  \\ [0.5ex] 
 \hline\hline
 $A$ & 0 & 1 \\ 
 \hline
 $B$ & 100 & 1  \\
 \hline
\end{tabular}
\caption{Counterexample for perfect Envy-freeness}

\label{table 5}
\end{table}

This motivates taking the utility of tenants into account in defining envy. However, this quickly gets tricky, since the utility can be a complex function of the entire configuration, and not of individual views. In order to develop more insight, we use an analogy to public projects. The tenants are members of a society, who contribute equal amount of tax. The total tax is the cache space. Each view is a public project whose cost is equal to its size. Users derive utility from the subset of projects built (or views cached). In a societal context, users are envious if they perceive an inordinate fraction of tax dollars being spent on making a small number of users happy. In other words, if they perceive a {\em bridge to nowhere} being built. Let us revisit the example in Table~\ref{table 4}. Here, the MMF allocation sets $\mathbf{x}=<x_{R}=\frac{1}{2}, x_{S}=\frac{1}{2}>$ and ignores the fact that an arbitrarily large number of tenants want R, compared to just one tenant who wants S. If we treat $R$ as a school and $S$ as a park, an arbitrarily large number of users want a school compared to a park, yet half the money is spent on the school, and half on the park. This will be perceived as unfair on a societal level.  

\paragraph{Randomized Core}
In order to formalize this intuition, we borrow the notion of core from cooperative game theory and exchange market economics \cite{gillies,core,continuum,exchange}. We treat each user as bringing a {\em rate endowment} of $\frac{1}{N}$ to the system. If they were the only user in the system, we would produce an allocation $\mathbf{x}$ with $||\mathbf{x}|| = \frac{1}{N}$ and maximize their utility. An allocation $\mathbf{x}$ over all tenants lies in the {\em core} if no subset of tenants can deviate and obtain better utilities for all participants by pooling together their rate endowments. More formally,

\begin{newdefn}
An allocation $\mathbf{x}$ is said to lie in the (randomized) \textbf{core} if for any subset $T$ of $N$ tenants, there is no feasible allocation $\mathbf{y}$ such that $||\mathbf{y}|| = \frac{|T|}{N}$, for which $U_i(\mathbf{y}) \geq U_i(\mathbf{x}), \forall i \in T$ and $U_j(\mathbf{y}) > U_j(\mathbf{x})$ for at least one $j \in T$.
\end{newdefn}

It is easy to check that any allocation in the core is both SI and PE, by considering sets $T$ of size $1$ and $N$ respectively. In the above example (Table~\ref{table 4}), the allocation $\mathbf{x}=<x_{R}=\frac{N-1}{N}, x_{S}=\frac{1}{N}>$ lies in the core. Tenant $T_N$ gets its SI amount of utility and cache space. The more demanded view $R$ is cached by a proportionally larger amount. In societal terms, each user perceives his tax dollars as being spent fairly.  Similarly, in the example in Table~\ref{table 5}, the allocation $\mathbf{x}=<x_{R}=\frac{1}{2}, x_{S}=\frac{1}{2}>$ lies in the core. 


In the context of provisioning public goods, there are two solution concepts that are known to lie in the core: The first, termed a Lindahl equilibrium~\cite{foley,mas-colell} attempts to find per-tenant prices that implement a Walrasian equilibrium, while the second, termed ratio equilibrium~\cite{kaneko} attempts to find per-tenant ratios of cache-shares. However, these concepts are shown to exist using fixed-point theorems, which don't lend themselves to efficient algorithmic implementations. We sidestep this difficulty by using randomization to our advantage, and show that a simple mechanism finds an allocation in the core.


\paragraph{Proportional Fairness}
\begin{newdefn}
An allocation $\mathbf{x}$ is \textbf{proportionally fair} (PF)  if it is a solution to:
\begin{equation}
\mbox{Maximize} \sum_{i=1}^{N} \log(U_i(\mathbf{x})) \mbox{ subject to: } ||\mathbf{x}|| \le 1
\label{equation 1}
\end{equation}
\label{definition 1}
\end{newdefn} 

We show the following theorem using the KKT (Karush-Kuhn-Tucker) conditions~\cite{kkt}. The proof also follows easily from the classic first order optimality condition of PF~\cite{bargaining}; however, we present the entire proof for completeness. Subsequently, in Section~\ref{sec:algo}, we show how to compute this allocation efficiently.

\begin{theorem}
\label{thm:pf}
Proportionally fair allocations satisfy the core property.
\end{theorem}
\begin{proof}
 Let $\mathbf{x}$ denote the optimal solution to (PF). Let $d$ denote the dual variable for the constraint $||\mathbf{x}|| \le 1$. By the KKT conditions, we have:
 \begin{displaymath}
x_S > 0 \implies \sum_i \frac{U_i(S)}{U_i(\mathbf{x})} = d 
 \end{displaymath}
 \begin{displaymath}
x_S = 0 \implies \sum_i \frac{U_i(S)}{U_i(\mathbf{x})} \le d  
 \end{displaymath}
 Multiplying the first set of identities by $x_S$ and summing them, we have
 $$d = d(\sum_S x_S) = \sum_i \frac{\sum_S x_S U_i(S)}{U_i(x)} = \sum_i 1 = N$$ 

This fixes the value of $d$. Next, consider a subset $T$ of users, with $|T| = K$, along with some allocation $\mathbf{y}$ with $||\mathbf{y}|| = \frac{K}{N}$. First note that the KKT conditions implied:
$$ \sum_i \frac{U_i(S)}{U_i(\mathbf{x})} \le N \ \ \forall S$$
Multiplying by $y_S$ and summing, we have:
$$ \sum_i \frac{U_i(\mathbf{y})}{U_i(\mathbf{x})}  \le N \sum_S y_S = K$$
Therefore, 
 \begin{displaymath}
 \sum_{i \in T} \frac{U_i(\mathbf{y})}{U_i(\mathbf{x})} \le K
 \end{displaymath}
Therefore, if $U_i(\mathbf{y}) > U_i(\mathbf{x})$ for some $i  \in T$, then there exists $j \in T$ for which  $U_j(\mathbf{y}) < U_j(\mathbf{x})$. This shows that no subset $T$ can deviate to improve their utility, so that the (PF) allocation lies in the core.
\end{proof}

\subsection{Discussion} 
\label{sec:discuss-fair}
Our notion of core easily extends to tenants having weights. Suppose tenant $i$ has weight $\lambda_i$. Then an allocation $\mathbf{x}$ belongs to the core if for all subsets $T$ of tenants, there does not exist $\mathbf{y}$ with $||\mathbf{y}|| = \frac{\sum_{i \in T} \lambda_i}{\sum_{i=1}^N \lambda_i}$ such that for all tenants $i \in T$, $U_i(\mathbf{x}) \le U_i(\mathbf{y})$, and $U_j(\mathbf{x}) < U_j(\mathbf{y})$ for at least one $j \in T$. The proportional fairness algorithm is modified to maximize $\sum_i \lambda_i \log U_i(\mathbf{x})$ subject to $||\mathbf{x}|| \le 1$.

We note that the PF algorithm finds an allocation in the (randomized) core to any resource allocation game that can be specified as follows: The goal is to choose a randomization over feasible configurations of resources. Each configuration yields an arbitrary utility to each tenant. This model is fairly general. For instance, consider the setting in~\cite{drf,beyonddrf}, where resources can be partitioned fractionally between agents, and an agent's rate (utility) depends on the minimum resource requirement satisfied in any dimension. Suppose we treat each agent as being endowed with $\frac{1}{N}$ fraction of the supply of resources in all dimensions, the above result shows that the (PF) allocation satisfies the property that no subset of users can pool their endowments together to achieve higher rates for all participants. 

\paragraph{Utilities under MMF and PF}
We now compare the total utility, $\sum_i V_i(\mathbf{x})$ for the optimal MMF and (PF) solutions. 
We  present results showing that (PF) has larger utility than MMF in certain canonical scenarios. Our first scenario defines the following {\em grouped} instance: There are $k$ views, $1,2,\ldots,k$ each of unit size. The cache also has size $1$. There are $k$ groups of tenants; group $i$ has $N_i$ tenants all of which want view $i$.

\begin{lemma}
The total utility of (PF) is at least the total utility of MMF for any grouped instance.
\end{lemma}
\begin{proof}
On grouped instances, MMF sets rate $1/k$ for each tenant, yielding a total utility of $N/k$ for $N$ tenants. The (PF) algorithm sets rate $x_i = N_i/N$ for all tenants in group $i$. This yields total utility of $\sum_i N_i^2/N$. Next note that
$$ \frac{\sum_{i=1}^k N_i^2}{k} \ge \left(\frac{\sum_{i=1}^k N_i}{k}\right)^2$$
Noting that $\sum_i N_i = N$, it is now easy to verify that (PF) yields larger utility.
\end{proof}

 In fact, the ratio of the utilities of MMF and PF is precisely the Jain's index~\cite{jain} of the vector $\langle N_1,N_2,\ldots,N_k \rangle$. By setting $k = N/2+1$, and $N_2 = N_3= \cdots = N_k = 1$, this shows that (PF) can have $\Omega(N)$ times larger total utility than MMF.  Our next scenario focuses on arbitrary instances with only two tenants.

\begin{lemma}
For two tenants, the total utility of (PF) is at least the total utility of MMF.
\end{lemma}
\begin{proof}
Let the utilities of the two tenants be $a,b$ in (PF) and $A,B$ in MMF. Assume $a \le b$. Since MMF maximizes the minimum utility, we have $a \le  \min(A,B)$. Let $\alpha = A/a$ and $\beta = B/b$, so that $\alpha \ge 1$. Since $\log{(a)} + \log{(b)} = \log{(ab)}$ is maximized by definition of PF and $\log$ is an increasing function, we have $ab \ge AB$, so $\alpha \beta \le 1$. Since $\alpha \ge 1$, this implies $1/\beta \ge \alpha \ge 1$. Therefore
$$ b - B = B\left(1/\beta  - 1\right) \ge a\left(1/\beta - 1\right) \ge a\left(\alpha - 1\right) = A - a$$
This shows $a + b \ge A + B$ completing the proof.
\end{proof}

\paragraph{Summary of Fairness Properties}
In summary, Table~\ref{table 6} shows the fairness properties that hold for all of our candidate algorithms.  We abbreviate the properties SI for sharing incentive and PE for pareto efficiency.  Based on this analysis, we suggest that proportional fairness is likely to be a preferable view selection algorithm for our ROBUS framework.  The theoretical properties of proportional fairness suggest that it should perform fairly and efficiently. 

\begin{table}[h!]
\centerline{
 \begin{tabular}{|c ||c|c|c|c|} 
 \hline
 Algorithm & SI & PE & CORE \\ [0.5ex] 
 \hline\hline
 Random Serial Dictatorship & \checkmark &  &\\
 \hline
Utility Maximization & & \checkmark & \\  
 \hline
 Max-Min Fairness & \checkmark & \checkmark &  \\
 \hline
 Proportional Fairness & \checkmark & \checkmark & \checkmark \\ [1ex]
 \hline
\end{tabular}}
\caption{Fairness properties of mechanisms}

\label{table 6}
\end{table}

%% file: algos.tex
\section{Approximately Computing PF and MMF Allocations} 
\label{sec:algo}

In this section, we show that the PF and MMF allocations can be computed to arbitrary precision. We then present fast heuristic algorithms for approximately computing PF and MMF allocations, which we implement in our prototype.

One key issue in computation is that the number of configurations is exponential in the number of views and tenants, so that the convex programming formulations have exponentially many variables. Nevertheless, since the programs have $O(N)$ constraints, we use the multiplicative weight method~\cite{ahk,Freund-Schapire} to solve them approximately in time polynomial in $N$ and accuracy parameter $1/\epsilon$. These algorithms assume access to a {\em welfare maximization} subroutine that we term {\sc Welfare}. 

\begin{newdefn}
Given weight vector $\mathbf{w}$, {\sc Welfare}$(\mathbf{w})$ computes a configuration $S$ that maximizes weighted scaled utilities, {\em i.e.}, solves $arg\,max_S \sum_{i=1}^{N} w_i V_i(S)$.
\end{newdefn}

The scaled utilities are computed using the tenant utility model described in Section~\ref{sec:archi}. In our presentation, we assume {\sc Welfare} solves the welfare maximization problem exactly. Our algorithms will make polynomially many calls to {\sc Welfare}. 

\paragraph{Multiplicative Weight Method}
\label{sec:ahk}
We first detail the multiplicative weight method, which will serve as a common subroutine to all our provably good algorithms. This classical framework~\cite{ahk,Freund-Schapire} uses a Lagrangian update to decide feasibility of linear constraints to arbitrary precision. 

We first define the generic problem of deciding the feasibility of a set linear constraints: Given a convex set $P\in \mathbf{R}^s$, and an $r \times s$ matrix $A$,

\begin{center}
\fbox{\begin{minipage}{2.6in}
\centerline{{\sc LP}$(A,b,P)$:  $\exists x \in P$ such that $A x \ge b$?}
\end{minipage}}
\end{center}

Let $y \ge 0$ be an $r$ dimensional dual vector for the constraints $A x \ge b$. We assume the existence of an efficient {\sc Oracle} of the form: 

\begin{center}
\fbox{\begin{minipage}{2.5in}
\centerline{ {\sc Oracle} $C(A, y) = \max\{ y^t A z : z \in P\}$.}
\end{minipage}}
\end{center}

The {\sc Oracle} can be interpreted as follows: Suppose we take a linear combination of the rows of $Ax$, multiplying row $a_i x$ by $y_i$. Suppose we maximize this as a function of $x \in P$, and it turns out to be smaller than $y^T b$. Then, there is no feasible way to satisfy all constraints in $Ax \ge b$, since the feasible solution $x$ would make $y^T A x \ge y^T b$. On the other hand, suppose we find a feasible $x$. Then, we check which constraints are violated by this $x$, and increase the dual multipliers $y_i$ for these constraints. On the other hand, if a constraint is too slack, we decrease the dual multipliers. We iterate this process until either we find a $y$ which proves $A x \ge b$ is infeasible, or the process roughly converges. 

More formally, we present the Arora-Hazan-Kale (AHK) procedure~\cite{ahk} for deciding the feasibility of {\sc LP}$(A,b,P)$.  The running time is quantified in terms of the {\sc Width} defined as:
$$\rho = \max_i \max_{x \in {P}} |a_i x - b_i|$$

\begin{algorithm}[htbp]
\caption{AHK Algorithm}
\begin{algorithmic}[1]
\State Let $K \leftarrow  \frac{4\rho^2 \log r}{\delta^2}$; $y_1 = \vec{1}$
\For{$t = 1$ {\bf to} $K$}
\State Find $x_t$ using {\sc Oracle} $C(A, y_t)$.
\If{$C(A, y_t) < y_t^T b$}
\State Declare {\sc LP}$(A,b,P)$ infeasible and terminate.
\EndIf
\For{$i = 1$ {\bf to} $r$}
\State $M_{it} = a_{i} x_t - b_{i}$  \Comment{Slack in constraint $i$.}
\State $y_{it+1} \leftarrow y_{it}(1 - \delta)^{M_{it}/\rho}$ if  $M_{it} \ge 0$. 
\State $y_{it+1} \leftarrow y_{it}(1 + \delta)^{-M_{it}/\rho}$ if  $M_{it} < 0$. 
\State \Comment{Multiplicatively update $y$.}
\EndFor
\State  Normalize $y_{t+1}$ so that $||y_{t+1}|| = 1$. 
\EndFor
\State {\bf Return} $x = \frac{1}{K} \sum_{t=1}^K x_t$.
\end{algorithmic}
\end{algorithm}

This procedure has the following guarantee~\cite{ahk}:
\begin{theorem}
If  {\sc LP}$(A,b,P)$ is feasible, the AHK procedure never declares infeasibility, and the final $x$  satisfies:
$$ (a_{i} x - b_{i})  + \delta \ge 0 \qquad \forall i$$
\end{theorem}

\subsection{Proportional Fairness}
\label{sec:pfalg}

Our algorithm uses the AHK algorithm as a subroutine and considers dual weights to find an additive $\epsilon$ approximation solution. 
The primary result is the following theorem:

\begin{theorem}
An approximation algorithm computes an additive $\epsilon$ approximation to (PF) with $O(\frac{4 N^4 \log^2 N}{\epsilon^2})$ calls to {\sc Welfare}, and polynomial additional running time.
\label{theorem 4}
\end{theorem}

\paragraph{Proof}
For allocation $\mathbf{x}$, let $B(\mathbf{x}) = \sum_i \log V_i(\mathbf{x})$. Let $Q^* = \max_{\mathbf{x}} B(\mathbf{x})$ denote the optimal value of (PF), and let $\mathbf{x}^*$ denote this optimal value. We first present a Lipschitz type condition, whose proof we omit from this version. 

\begin{lemma}
Let $\mathbf{y}$ satisfy $B(\mathbf{y}) \ge Q^* - \epsilon$ for $\epsilon \in (0,1/6)$. Then, for all $i$, $V_i(\mathbf{y}) \ge V_i(\mathbf{x})/2$.
\end{lemma}

The proof idea is to use the concavity of the log function to exhibit  a convex combination of $\mathbf{x}$ and $\mathbf{y}$ whose value exceeds $Q^*$, which is a contradiction. It is therefore sufficient to find $Q^*$ to an additive approximation in order to achieve at least half the welfare of (PF) for all tenants. Towards this end, for a parameter $Q$, we write (PF) as a feasibility problem {\sc PFFeas}$(Q)$ as follows:

\begin{newdefn}
{\sc PFFeas}$(Q)$ decides the feasibility of the constraints
$$ (F) = \left\{ \sum_S x_S V_i(S) - \gamma_i \ge 0  \ \forall i\right\}$$
subject to the constraints:
$$ (P1) = \left\{\sum_S x_S \le 1, \ x_S \ge 0 \ \forall S \right\}$$
$$ (P2) = \left\{\sum_i \log \gamma_i \ge Q, \ \gamma_i \in [1/N,1] \ \forall i \right\}$$
\end{newdefn}

The above formulation is not an obvious one, and is related to {\em virtual welfare} approaches recently proposed in Bayesian mechanism design~\cite{CaiDW,BhalgatGM13}. The key idea is to connect expected values (utility) to their realizations in each configurations via expected value variables, the $\gamma_i$. The constraints (P2) and (P1) are over expected values, and realizations respectively. The {\sc Oracle} computation in the multiplicative weight procedure will decouple into optimizing expected value variables over (P2), and optimizing {\sc Welfare} over (P1) respectively, and both these problems will be easily solvable. 

We note that (P2) has additional constraints $\gamma_i \in [1/N,1] \ \forall i$. These are in order to reduce the width of the constraints (F). Note that otherwise, $\gamma_i$ can take on unbounded values while still being feasible to (P2), and this makes the width of (F) unbounded. The lower bound of $1/N$ on $\gamma_i$ is to control the approximation error introduced. We argue below that these constraints do not change our problem.

\begin{lemma}
Let $Q^*$ denote the optimal value of the proportional fair allocation (PF). Then, {\sc PFFeas}$(Q)$ is feasible if and only if $Q \le Q^*$.
\end{lemma}
\begin{proof}
In the formulation {\sc PFFeas}$(Q)$, the quantity $\gamma_i$ is simply the scaled utility of tenant $i$. Consider the proportionally fair allocation $\mathbf{x}$. For this allocation, all scaled utilities lie in $[1/N,1]$ since the allocation is SI. Therefore, $\mathbf{x}$ is feasible for {\sc PFFeas}$(Q^*)$. On the other hand, if $\mathbf{y}$ is feasible to {\sc PFFeas}$(Q)$ for $Q > Q^*$, then $\mathbf{y}$ is also feasible for (PF), contradicting the optimality of $\mathbf{x}$.
\end{proof}

We will therefore search for the largest $Q$ for which {\sc PFFeas}$(Q)$ is feasible. Since each $\gamma_i \in [1/N,1]$, we have $Q \in [-N \log N, 0]$. Therefore, obtaining an additive $\epsilon$ approximation to $Q^*$ by binary search requires $O(\log N)$ evaluations of {\sc PFFeas}$(Q)$ for various $Q$, assuming constant $\epsilon > 0$.\\

\noindent {\bf Solving {\sc PFFeas}$(Q)$.} We now fix a value $Q$ and apply the AHK procedure to decide the feasibility of {\sc PFFeas}$(Q)$. To map to the description in the AHK procedure, we have $b = 0$, and $A$ is the LHS of the constraints (F). We have $r = N$. Since any $V_i(S) \le 1$, and $\gamma_i \in [1/N,1]$, the width $\rho$ of (F) is at most $1$. Finally, for small constant $\epsilon > 0$, we will set $\delta = \frac{\epsilon}{N^2}$. Therefore, $K = \frac{4 N^4 \log N}{\epsilon^2}$. 

For dual weights $\mathbf{w}$, the oracle subproblem $C(A,\mathbf{w})$ is the following:
$$ \mbox{Max}_{\mathbf{x},\gamma} \sum_i \left(w_i V_i(\mathbf{x}) -  \gamma_i\right)$$
subject to (P1) and (P2). This separates into two optimization problems.

The first sub-problem maximizes $\sum_i w_i V_i(\mathbf{x})$ subject to $\mathbf{x}$ satisfying (P1). This is simply {\sc Welfare}$(\mathbf{w})$. The second sub-problem is the following:
$$ \mbox{Minimize} \sum_i w_i \gamma_i$$
subject to $\mathbf{w}$ satisfying (P2). Let $L$ denote the dual multiplier to the constraint $\sum_i \log \gamma_i \ge Q$. Consider the Lagrangian problem:
$$ \mbox{Minimize} \sum_i (w_i \gamma_i - L \log \gamma_i)$$ 
subject to $\gamma_i \in [1/N,1]$ for all $i$. The optimal solution sets $\gamma_i(L) = \max(1/N, \min(1,L/w_i))$, which is an non-decreasing function of $L$. We check if $\sum_i \gamma_i(L)  < Q$. If so, we increase $L$ till we satisfy the constraint with equality. This parametric search takes polynomial time, and solves the second  sub-problem.

\medskip
The AHK procedure now gives the following guarantee: Either we declare {\sc PFFeas}$(Q)$ is infeasible, or we find $(\mathbf{x},\gamma)$ such that for all $i$, we have:
$$ \sum_S x_S V_i(S) \ge \gamma_i - \epsilon/N^2 \ge \gamma_i(1-\epsilon/N)$$

Since $\sum_i \log \gamma_i \ge Q$, the above implies: 
$$B(\mathbf{x}) = \sum_i \log V_i(\mathbf{x}) \ge Q - \sum_i \log(1-\epsilon/N) \ge Q - \epsilon$$
so that the value $Q-\epsilon$ is achievable with the allocation $\mathbf{x}$. \\

\noindent {\bf Binary Search.} To complete the analysis, since {\sc PFFeas}$(Q^*)$ is feasible, the procedure will never declare infeasibility when run with $Q = Q^*$, and will find an $\mathbf{x}$ with $B(\mathbf{x}) \ge Q^* - \epsilon$, yielding an additive $\epsilon$ approximation. This binary search over $Q$ takes $O(\log N)$ iterations. \\

Thus, we arrive at the result of theorem \ref{theorem 4}.

\subsection{Max-min Fairness}
\label{sec:mmf-alg}
We present an algorithm {\sc SimpleMMF} that computes an allocation $\mathbf{x}$ maximizing $\min_i V_i(\mathbf{x})$. The MMF allocation can be computed by applying this procedure iteratively as in~\cite{choosy}; we omit  the simple details from this version. We note that the idea of applying the multiplicative weight method to compute max-min utility also appeared in~\cite{vazirani}. 

\noindent We write the problem of deciding feasibility as {\sc SimpleMMF}$(\lambda)$:
$$ (F) = \left\{ \sum_S V_i(S) x_S \ge \lambda \ \forall i \right\}$$
subject to the constraints:
$$ (P) = \left\{\sum_S x_S \le 1, \ x_S \ge 0 \ \forall S \right\}$$
We have $\lambda^* \in [1/N,1]$, where $\lambda^* = \max_{\mathbf{x}} \min_i V_i(\mathbf{x})$. Therefore, the width $\rho \le 1$. Further, we can set $\delta = \epsilon/N$. We can now compute $K$ from the AHK procedure, so that $K = \frac{4N^2 \log N}{\epsilon^2}$ in order to approximate $\lambda^*$ to a factor of $(1-\epsilon)$. The procedure is described in Algorithm~\ref{alg2}. 

\begin{algorithm}
\caption{\label{alg2}Approximation Algorithm for {\sc SimpleMMF}}
\begin{algorithmic}[1]
\State Let $\epsilon$ denote a small constant $<$ 1.
\State $T\gets \frac{4N^2 \log N}{\epsilon^2}$
\State $\mathbf{w_1}\gets {\frac{1}{N}}$ \Comment{Initial weights}
\State $\mathbf{x}\gets \mathbf{0}$ \Comment{Probability distribution over set of views}
\For {$k\in 1,2,\ldots,T$}
 \State Let $S$ be the solution to {\sc Welfare}$(\mathbf{w_k})$.
 \State $w_{i(k+1)}\gets w_{ik}  \exp(-\epsilon \frac{U_i(S)}{U_i^*})$
 \State Normalize  $\mathbf{w_{k+1}}$ so that $||\mathbf{w_{k+1}}|| = 1$.
 \State $x_S\gets x_S + \frac{1}{T}$ \Comment{Add $S$ to collection}
\EndFor
\label{mmf}
\end{algorithmic}
\end{algorithm}

In order to compute MMF allocations, we use a similar idea to decide feasibility, except that we have to perform $O(N^2)$ invocations. This blows up the running time to $O\left(\frac{4N^4 \log N}{\epsilon^2}\right)$ invocations of {\sc Welfare}.

The algorithm gives the following result:

\begin{theorem}
An approximation algorithm for {\sc SimpleMMF} (Algorithm~\ref{alg2}) finds a solution $\mathbf{x}$ such that $\min_i V_i(\mathbf{x}) \ge \lambda^* (1-\epsilon)$ using $\frac{4N^2 \log N}{\epsilon^2}$ calls to {\sc Welfare}.
\label{claim 2}
\end{theorem}

\renewcommand{\S}{\mathcal{S}}
\subsection{Fast Heuristics}
\label{sec:heuristic}

In this section, we present heuristic algorithms that directly work with the exponential size convex programs. We directly implement these algorithms in software to gather our experimental results. 

\paragraph{Configuration Pruning}  For $M = O(N^2)$, generate $M$ random $N$-dimensional unit vectors $w_k, k = 1,2,\ldots, M$. For each $w_k$, let $S_k$ be the configuration corresponding to {\sc Welfare}$(w_k)$. Denote this set of configurations by $\S$. We restrict the convex programming formulations of PF and MMF to just the set of configurations $\S$, and solve these programs directly, as we describe below.  The intuition behind doing this pruning step is the following: The approximation algorithms for PF and MMF find convex combinations of configurations that are optimal for {\sc Welfare}$(w)$ for some $w$'s that are computed by the multiplicative weight procedure. Instead of this, we generate random such Pareto-optimal configurations, giving sufficient coverage so that each tenant has a high probability of having the maximum weight at least once.

We compared two algorithms for {\sc SimpleMMF}, one using the multiplicative weight procedure (Algorithm~\ref{alg2}), and the other solving the linear program (Program (\ref{eq3}) below) restricted to random optimal configurations. When run on 200 batches with five tenants, using 5 weight vectors gives a 10.4\% approximation to the objective of {\sc SimpleMMF}.  With 25 random weight vectors, the approximation error is 1.4\%, and using 50 random weights, the approximation error drops to 0.6\%. This shows that a small set $\S$  of configurations that are optimal solutions to {\sc Welfare}$(w)$ for random vectors $w$ is sufficient to generate good approximations to our convex programs. In our implementation, we set $\S$ to be the union of  these configurations along with the configurations generated by the {\sc SimpleMMF} algorithm (Algorithm~\ref{alg2}). 

\paragraph{Proportional Fairness} We first note that (PF) is equivalent to the following; the proof of equivalence follows from Theorem~\ref{thm:pf}, where the dual variable corresponding to the constraint $\sum_S x_S = 1$ is precisely $N$.
\begin{equation}
\mbox{Max } g(\mathbf{x}) = \sum_{i=1}^{N} \log(V_i(\mathbf{x})) - N \|\mathbf{x}\| \ \ \mbox{ s.t.: } \mathbf{x} \geq 0
\label{equation 2}
\end{equation}
Given a configuration space $\S$, we can solve the program (\ref{equation 2}) using gradient descent, as shown in Algorithm~\ref{alg3}. As precomputation, for each configuration $S \in \S$, we precompute $V_i(S)$. Then $V_i(\mathbf{x}) = \sum_{S \in \S} V_i(S) x_S$. 

\begin{algorithm}[htbp]
\caption{\label{alg3}Proportional Fairness Heuristic}
\label{algo:fpf}
\begin{algorithmic}[1]
\State Let $M = |\S|$. Set $t = 1$.
\State Let $\mathbf{x}_1 = (1/M,1/M,\ldots,1/M)$.
\Repeat
\State $\mathbf{y} = \nabla g(\mathbf{x})$ evaluated at $\mathbf{x} = \mathbf{x}_t$.
\State $r^* = arg\,max_r\left( g(\mathbf{x}_t+r\mathbf{y}) \right)$
\State $\mathbf{x}_{t+1} = \mathbf{x}_t+r^*\mathbf{y}$
\State Project $\mathbf{x}_{t+1}$ as: $x_d = \max(x_d, 0)$ for all dimensions $d \in \{1,2,\ldots,M\}$.
\Until{$\mathbf{x}_t$ converges}
\end{algorithmic}
\end{algorithm}

\paragraph{Max-min Fairness} Using the precomputed configuration space $\S$, we solve {\sc SimpleMMF} using the following linear program:
\begin{equation}
\label{eq3}
\max\left\{\lambda \ | \ \sum_{S \in \S} V_i(S) x_S \ge \lambda \ \forall i, \mathbf{x} \ge \mathbf{0} \right\}
\end{equation}
This can be solved using any off-the-shelf LP solver (our implementation uses the open source lpsolve package \cite{lpsolve}). In order to compute the MMF allocation, we iteratively compute the lexicographically max-min allocation using the above LP. The details are standard; see for instance~\cite{choosy}. Briefly, in each iteration a value of $\lambda$ is computed.  All tenants whose rate cannot be increased beyond $\lambda$ without decreasing the rate of another tenant are considered saturated and the rate of $\lambda$ for these tenants is a constraint in the next iteration of the LP.  The solution to the final LP for which all tenants are saturated is the MMF solution.

%% file: evaluation.tex
\section{Evaluation} 
\label{sec:eval}

 We evaluate cache allocation policies on a variety of practical setups of multi-tenant analytics clusters. The setups may differ in the number of tenants, workload arrival patterns, data access patterns, etc. Some of the example setups are listed below. 

\squishlist

 \item{Analysts}: Tenants correspond to various BI analysts in an enterprise that run a similar workload. Some of the datasets are frequently accessed by all tenants suggesting a good opportunity for shared optimization.

 \item{ETL+Analysts}: All analysts have similar data access patterns as above. But additionally, a tenant runs ETL workload that may touch different datasets than the BI tenants. 


 \item{Production+Engineering}: Engineering workload is of bursty nature. Depending on the time of day, engineering queues have different amounts of work whereas production queues, running pre-scheduled workflows, have similar amounts of work.

\squishend

We replicate various combinations of these setups on a small-scale Spark cluster and run controlled experiments using a mix of TPC-H benchmark~\cite{tpch} workload and a synthetically generated scan-based workload. 

\subsection{Setup and Methodology}
\label{sec:setup}

Figure~\ref{fig:architecture} has presented the architecture of \robus. We use Apache Spark~\cite{spark} to build a system prototype. Spark is a natural choice for the evaluation since it supports distributed memory-based abstraction in the form of Resilient Distributed Datasets (RDDs). 
In our prototype, a long running Spark context is shared among multiple queues, with each queue corresponding to a tenant. The Spark context has an access to the entire RDD cache in the cluster. Spark's internal fair share scheduler is configured with a dedicated pool for each queue; the fair share properties of the pool are set proportional to weight of the corresponding queue. 

\begin{table}[h!]
\centering
 \begin{tabular}{|l | c |} 
 \hline
 Spark version & 1.1.1 \\
 \hline
 Number of worker nodes & 10 \\ 
 \hline
 Instance type of nodes & c3.2xlarge \\
 \hline
 Total number of cores & 80 \\
 \hline
 Executor memory & 80GB \\ 
 \hline
\end{tabular}

\caption{Test cluster setup on Amazon EC2}

\label{tab:setup}
\end{table}

Table~\ref{tab:setup} presents our test cluster setup. We generate two types of data to reflect two types of uses observed in typical multi-tenant clusters: (a) A set of 30 datasets with varying sizes each matching schema of the ``sales'' tables---{$store\_sales$, $catalog\_sales$, and $web\_sales$}---from TPC-DS benchmark~\cite{tpcds} data, and (b) All TPC-H benchmark~\cite{tpch} datasets generated at scale 5. 

The first category of data represents raw fact/log data that comes
into the cluster from the OLTP/operational databases in a company.
This data is processed by synthetically-generated ETL and exploratory
SQL queries, each performing scans and aggregations over a dataset. We
refer to this category of queries as the {\sf Sales} workload. Total
size of {\sf Sales} data on disk is 600GB. We create a vertical
projection view on each dataset on its most frequently accessed
columns and use it whenever possible to answer queries. Sizes of these
views when loaded to cache range from 118MB to 3.6GB as can be seen in Figure~\ref{fig:datasizes}.

The second category represents data in the cluster after it has been
processed by ETL. Note that this data is typically much smaller in
size compared to fact/log data. In our experiments, this data is
queried by standard {\sf TPC-H} benchmark queries which consist of a
suite of business-oriented analytics and involves more complex
operations, such as joins, compared to the {\sf Sales} workload. All
of the queries in our evaluation are submitted using SparkSQL APIs.

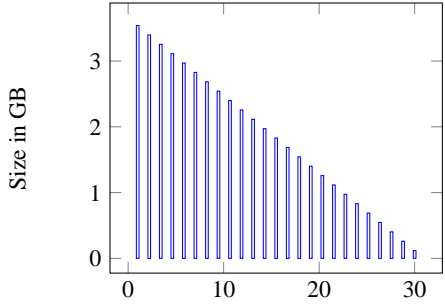
\begin{figure}
\pgfplotsset{width=6cm}
\centering

\begin{tikzpicture}
\begin{axis}[
ylabel=Size in GB,
]
\addplot[ybar, blue, bar width=1, domain=1:30] {.118*(31-x)};
{plots/dataset-sizes.dat};

\end{axis}
\end{tikzpicture}

\caption{Cache size estimates of candidate {\sf Sales} views}

\label{fig:datasizes}
\end{figure}

We set the cache size to 8GB, 10\% of the total executor memory, leaving aside the rest as a heap space. Only 6GB of the cache is used to carry out our optimizations in order to avoid memory management issues our Spark installation experienced while evicting from a near-full cache. 

The tenant utility model we use to estimate the utility of a cache configuration in our evaluation
 is based on the observations made from real-life clusters in~\cite{pacman}. If all the datasets that a query needs are cached, then the query is assigned a utility equal to the total size of data it reads; which corresponds to the savings in disk read I/O. Otherwise, we assign a utility of zero. It is observed in~\cite{pacman} that queries do not benefit much from in-memory caching if any part of their working set is not cached. 

 The workflow is described in Section~\ref{sec:archi} already. Here we want to add the fact that the cache update phase in our evaluation setup only marks datasets for caching or uncaching using Spark's cache directives. Spark lazily updates the cache when the first query requesting cached data from the batch is scheduled for execution.

\paragraph{Workload Arrival and Data Access}

\begin{figure}
\centering
\includegraphics[width=0.8\columnwidth]{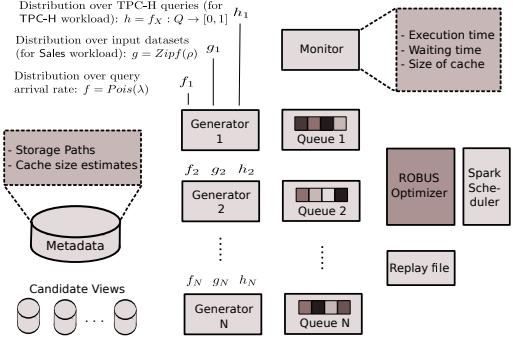}
\caption{Workload generation for \robus}
\label{fig:generation}

\end{figure}

Figure~\ref{fig:generation} shows our workload generation process. Several studies have established that query arrival times follow a Poisson distribution~\cite{gray, taobao}. We use the same in our prototype. 
Previous studies have also indicated that the data accessed by analytical workloads follows a Zipf distribution~\cite{gray, adolescence}: A small number of datasets are more popular than others, while there is a long tail of datasets that are only sporadically accessed. To replicate such data access, our synthetic {\sf Sales} workload generator picks a dataset from a Zipfian distribution provided at the time of configuration and adds grouping and aggregation predicates from a probability distribution defined for the chosen dataset. The {\sf TPC-H} workload generator, on the other hand, picks a benchmark query from a probability distribution over the queries provided at the time of configuration.

Further,~\cite{adolescence} also shows that 90\% of recently accessed data is re-accessed within next hour of first access. This makes a lot of sense because users typically want to drill down a dataset further in response to some interesting observation obtained in the previous run. In order to support such scenarios, we pick a small window in time from a Normal distribution. Over this window, a small subset of datasets is chosen from the Zipfian $g$. This subset forms candidates for the duration of the window. Each query to be generated picks one dataset from the candidates uniformly at random. This technique is taken from~\cite{gray} which terms the values used in local window as ``cold'' values to differentiate them from globally popular ``hot'' values. The generated workload still follows the Zipfian $g$ globally. The local distribution is optional; If not provided, datasets are picked from Zipfian $g$ at all times. 

\subsection{Performance Metrics}
\label{sec:metrics}
We gather several performance metrics while executing a workload.  They are defined next. We emphasize that these metrics are over long time horizons.
\begin{enumerate}
 \item {\bf Throughput.} This is simple to define:
 \begin{equation}
\mbox{Throughput }= \frac{\text{number of queries served}}{\text{total time taken}}
 \end{equation}

 \item {\bf Fairness Index.} 
 For job schedulers, a performance-based fairness index is defined in terms of variance in slowdowns of jobs in a shared cluster compared to a baseline case where every job receives all the resources~\cite{quincy}. As our work is about speeding up queries, we use relative speedups across queries while deriving fairness. The baseline is the case of statically partitioned cache. Here, $X_i$ is the mean speedup for tenant $i$, and $\lambda_i$ is the weight of tenant $i$.
 \begin{equation}
  \mbox{Fairness index }= \frac{(\sum_{i=1}^n \frac{X_i}{\lambda_i})^2}{n \sum_{i=1}^n ({\frac{X_i}{\lambda_i}})^2}
 \end{equation}
 
 \item {\bf Average Cache Utilization.}
This is simply the average fraction of cache utilized during workload execution.

 \item {\bf Hit Ratio.}
The fraction of queries served off cached views.

\end{enumerate}

Some of the other metrics we collect include flow time, mean execution time, mean wait time, and wait time fairness index. They are not included due to space constraints.

\subsection{Algorithm Evaluation}
\label{sec:tradeoffs}
In this section, we evaluate four view selection algorithms on various setups. Each algorithm processes a batch of query workload in an offline manner as detailed in Section~\ref{sec:archi}. Section~\ref{sec:model} discussed several possible algorithms. Here, we compare the following:
\begin{enumerate}
 \item \static: Cache is partitioned in proportion to weights of the tenants. We treat this as baseline when evaluating fairness index.
 \item \mmf: Max-min fairness implementation described in Section~\ref{sec:heuristic}.
 \item \fastpf: Proportional fairness implementation described in Section~\ref{sec:heuristic}.
 \item \opt: The only goal is to optimize for query performance; Workload from a batch is treated as if belonging to a single tenant -- a special case of either \mmf\ or \fastpf.
\end{enumerate}

In order to compare these algorithms across various settings, we vary the following parameters independently in our experiments.

\begin{enumerate}
 \item Data sharing among tenants (Section~\ref{sec:data});
 \item Workload arrival rate (Section~\ref{sec:arrival}); and
 \item Number of tenants (Section~\ref{sec:tenants}).
\end{enumerate}

\subsubsection{Effect of data sharing among tenants}
\label{sec:data}
To study the impact of different data sharing patterns on the performance of algorithms, we create four different workload distributions: $h_1$ picks queries uniformly at random over a set of 15 {\sf TPC-H} benchmark queries; $g_1-g_3$ create three different Zipf distributions over 30 {\sf Sales} datasets over which scan-and-aggregation queries are generated. Each of the distributions is skewed towards a different subset of datasets. Using these distributions, we create four test setups allowing different levels of data sharing, as listed in Table~\ref{tab:datasetup}. The batch size is set to 40 seconds; the inter-query arrival time distribution for all the tenants is given by Poisson(20); and we run 30 batches of workload for every data point.

\begin{table}
\centering

 \begin{tabular}{| c | c |} 
 \hline
 Setup & Distributions used by four tenants \\ 
 \hline
 $\mathcal{G}_1$ & \{$h_1$, $h_1$, $h_1$, $h_1$\} \\
 \hline
 $\mathcal{G}_2$ & \{$h_1$, $h_1$, $h_1$, $g_1$\} \\ 
 \hline
 $\mathcal{G}_3$ & \{$h_1$, $h_1$, $g_1$, $g_2$\} \\
 \hline
 $\mathcal{G}_4$ & \{$h_1$, $g_1$, $g_2$, $g_3$\} \\
 \hline
\end{tabular}

\caption{Data access distributions used in evaluation on a mixed workload}
\vspace*{-0.4cm}
\label{tab:datasetup}
\end{table}

\pgfplotscreateplotcyclelist{my black white}{%
solid, every mark/.append style={solid, fill=gray}, mark=*\\%
densely dashed, every mark/.append style={solid, fill=gray},mark=square*\\%
densely dotted, every mark/.append style={solid, fill=gray}, mark=diamond*\\%
dashdotted, every mark/.append style={solid, fill=gray},mark=star\\%
dotted, every mark/.append style={solid, fill=gray}, mark=square*\\%
loosely dotted, every mark/.append style={solid, fill=gray}, mark=triangle*\\%
loosely dashed, every mark/.append style={solid, fill=gray},mark=*\\%
dashdotted, every mark/.append style={solid, fill=gray},mark=otimes*\\%
dasdotdotted, every mark/.append style={solid},mark=star\\%
densely dashdotted,every mark/.append style={solid, fill=gray},mark=diamond*\\%
}

\begin{figure*}
\centering
\pgfplotsset{width=4.8cm}

\begin{tikzpicture}
\begin{axis}[
title=Throughput(/min),
symbolic x coords={1,2,3,4},
xticklabels={$\mathcal{G}_1$, $\mathcal{G}_2$, $\mathcal{G}_3$, $\mathcal{G}_4$},
enlargelimits=0.15,
xtick=data,
ymin=4,
ymax=20,
legend columns=-1,
legend entries={\static,\mmf,\fastpf,\opt},
legend to name=named,
cycle list name=my black white,
]
\addplot+[color=green, sharp plot] coordinates
{(1,7.8) (2,7.2) (3,7.2) (4,5.4)};
\addplot+[color=blue, sharp plot] coordinates
{(1,19.2) (2,9.0) (3,7.5) (4,5.4)};
\addplot+[color=red, sharp plot] coordinates
{(1,19.2) (2,10.2) (3,7.8) (4,5.4)};
\addplot+[color=black, sharp plot] coordinates
{(1,19.2) (2,16.2) (3,9.6) (4,4.8)};
\end{axis}

\end{tikzpicture}
\quad
\begin{tikzpicture}

\begin{axis}[
title=Average cache utilization,
ymin=0,
ymax=1,
symbolic x coords={1,2,3,4},
xticklabels={$\mathcal{G}_1$, $\mathcal{G}_2$, $\mathcal{G}_3$, $\mathcal{G}_4$},
enlargelimits=0.15,
xtick=data,
cycle list name=my black white,
]
\addplot+[color=green, sharp plot] coordinates
{(1,0) (2,.08) (3,.16) (4,.24)};
\addplot+[color=blue, sharp plot] coordinates
{(1,.83) (2,.81) (3,.96) (4,.91)};
\addplot+[color=red, sharp plot] coordinates
{(1,.83) (2,.87) (3,.98) (4,.93)};
\addplot+[color=black, sharp plot] coordinates
{(1,.83) (2,.92) (3,1.0) (4,.96)};
\end{axis}

\end{tikzpicture}
\quad
\begin{tikzpicture}

\begin{axis}[
title=Hit ratio,
ymin=0,
ymax=0.9,
symbolic x coords={1,2,3,4},
xticklabels={$\mathcal{G}_1$, $\mathcal{G}_2$, $\mathcal{G}_3$, $\mathcal{G}_4$},
enlargelimits=0.15,
xtick=data,
cycle list name=my black white,
]
\addplot+[color=green, sharp plot] coordinates
{(1,0) (2,.08) (3,.19) (4,.26)};
\addplot+[color=blue, sharp plot] coordinates
{(1,1) (2,.54) (3,.53) (4,.43)};
\addplot+[color=red, sharp plot] coordinates
{(1,1) (2,.68) (3,.55) (4,.47)};
\addplot+[color=black, sharp plot] coordinates
{(1,1) (2,.83) (3,.67) (4,.46)};
\end{axis}

\end{tikzpicture}
\quad
\begin{tikzpicture}

\begin{axis}[
title=Fairness index,
symbolic x coords={1,2,3,4},
xticklabels={$\mathcal{G}_1$, $\mathcal{G}_2$, $\mathcal{G}_3$, $\mathcal{G}_4$},
enlargelimits=0.15,
ymin=0.3,
xtick=data,
cycle list name=my black white,
]
\addplot+[color=green, sharp plot] coordinates
{(1,1) (2,1) (3,1) (4,1)};
\addplot+[color=blue, sharp plot] coordinates
{(1,.7) (2,.83) (3,.77) (4,.81)};
\addplot+[color=red, sharp plot] coordinates
{(1,.7) (2,.79) (3,.66) (4,.8)};
\addplot+[color=black, sharp plot] coordinates
{(1,.7) (2,.75) (3,.5) (4,.38)};
\end{axis}

\end{tikzpicture}
\\
\ref{named}

\caption{Effect of data sharing changes on four equi-paced tenants on a mixed workload}
\label{fig:tpchaccess}
\end{figure*}
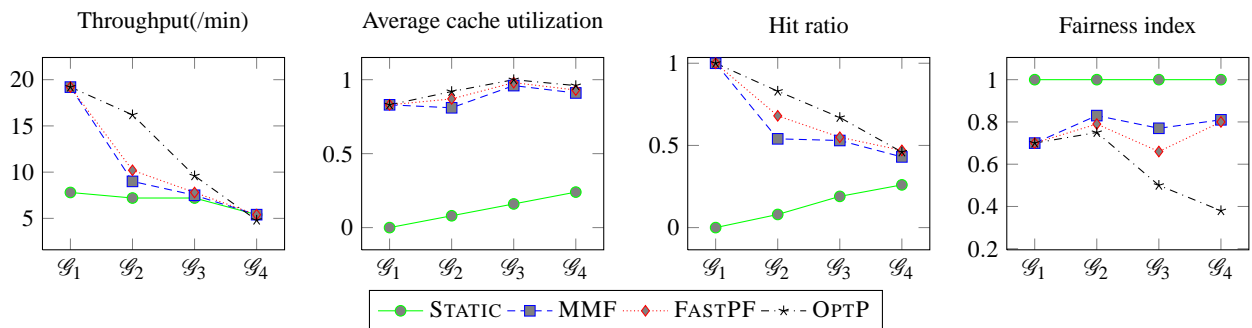

Figure~\ref{fig:tpchaccess} shows how different algorithms perform in each of these setups. Throughput goes down with heterogeneity in data access. \static\ policy fails to cache any dataset for {\sf TPC-H} workload because each of the queries we generate reads the largest table, {\em lineitem}, which amounts to $\approx 3.8GB$, much larger than cache at the disposal of \static. The other three policies, on the other hand, can serve every query off cache in setup $\mathcal{G}_1$. As a result, they exhibit a throughput of more than 2x over \static. However, as the heterogeneity in data access increases, the gap in throughput narrows. Even though the shared policies cache more data, frequent updates to cache configuration per batch cause additional delays. We explore possibility of retaining state of cache in Section~\ref{sec:evaldiscuss}.

Among the shared cache policies, \opt\ scores high on throughput but very low on fairness index. It uses cache exclusively for {\sf TPC-H} tenants at the cost of degradation to {\sf Sales} tenants' performance. \mmf\ and \fastpf\ policies, on the other hand show much better tradeoffs in terms of performance and fairness to tenants.

We also repeated the same experiment on {\sf Sales} data alone. We first create four different Zipf distributions over candidate views:  $g_1, g_2, g_3, g_4$. Each of the distributions is skewed towards a different subset of views.We create four test setups, each allowing a different level of data sharing, as listed in Table~\ref{tab:datasetupsales}. The other common parameters are listed in Table~\ref{tab:sdataaccess}. 

\begin{table}
\centering
 \begin{tabular}{| c | c |} 
 \hline
 Setup & Distributions used by four tenants \\ 
 \hline
 $\mathcal{G}_1$ & \{$g_1$, $g_1$, $g_1$, $g_1$\} \\
 \hline
 $\mathcal{G}_2$ & \{$g_1$, $g_1$, $g_1$, $g_2$\} \\ 
 \hline
 $\mathcal{G}_3$ & \{$g_1$, $g_1$, $g_2$, $g_3$\} \\
 \hline
 $\mathcal{G}_4$ & \{$g_1$, $g_2$, $g_3$, $g_4$\} \\
 \hline
\end{tabular}

\caption{Data access distributions used in evaluation on {\sf Sales} workload}

\label{tab:datasetupsales}
\end{table}

\begin{table}
\centering
 \begin{tabular}{| l | c |} 
 \hline
 Parameter & Value \\ 
 \hline
 Query inter-arrival rates (sec) & \{20 $\forall$ tenant\} \\
 \hline
 Batch size (sec) & 40 \\ 
 \hline
 Number of batches & 30 \\
 \hline
\end{tabular}

\caption{Data sharing experiment setup}

\label{tab:sdataaccess}
\end{table}

\pgfplotscreateplotcyclelist{my black white}{%
solid, every mark/.append style={solid, fill=gray}, mark=*\\%
densely dashed, every mark/.append style={solid, fill=gray},mark=square*\\%
densely dotted, every mark/.append style={solid, fill=gray}, mark=diamond*\\%
dashdotted, every mark/.append style={solid, fill=gray},mark=star\\%
dotted, every mark/.append style={solid, fill=gray}, mark=square*\\%
loosely dotted, every mark/.append style={solid, fill=gray}, mark=triangle*\\%
loosely dashed, every mark/.append style={solid, fill=gray},mark=*\\%
dashdotted, every mark/.append style={solid, fill=gray},mark=otimes*\\%
dasdotdotted, every mark/.append style={solid},mark=star\\%
densely dashdotted,every mark/.append style={solid, fill=gray},mark=diamond*\\%
}

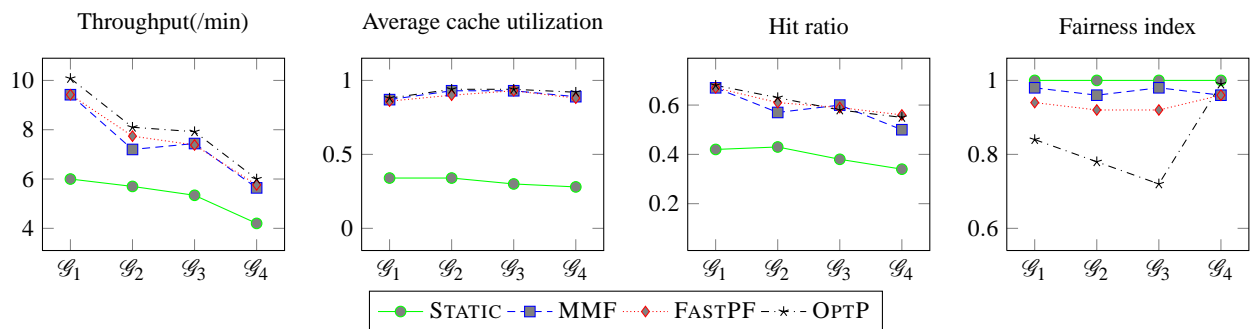
\begin{figure*}
\centering
\pgfplotsset{width=4.8cm}

\begin{tikzpicture}
\begin{axis}[
title=Throughput(/min),
symbolic x coords={1,2,3,4},
xticklabels={$\mathcal{G}_1$, $\mathcal{G}_2$, $\mathcal{G}_3$, $\mathcal{G}_4$},
enlargelimits=0.15,
xtick=data,
ymin=4,
ymax=10,
legend columns=-1,
legend entries={\static,\mmf,\fastpf,\opt},
legend to name=named,
cycle list name=my black white,
]
\addplot+[color=green, sharp plot] coordinates
{(1,6.0) (2,5.7) (3,5.34) (4,4.2)};
\addplot+[color=blue, sharp plot] coordinates
{(1,9.42) (2,7.2) (3,7.44) (4,5.64)};
\addplot+[color=red, sharp plot] coordinates
{(1,9.42) (2,7.74) (3,7.38) (4,5.76)};
\addplot+[color=black, sharp plot] coordinates
{(1,10.08) (2,8.1) (3,7.92) (4,6.0)};
\end{axis}

\end{tikzpicture}
\quad
\begin{tikzpicture}

\begin{axis}[
title=Average cache utilization,
ymin=0,
ymax=1,
symbolic x coords={1,2,3,4},
xticklabels={$\mathcal{G}_1$, $\mathcal{G}_2$, $\mathcal{G}_3$, $\mathcal{G}_4$},
enlargelimits=0.15,
xtick=data,
cycle list name=my black white,
]
\addplot+[color=green, sharp plot] coordinates
{(1,.34) (2,.34) (3,.30) (4,.28)};
\addplot+[color=blue, sharp plot] coordinates
{(1,.87) (2,.93) (3,.93) (4,.89)};
\addplot+[color=red, sharp plot] coordinates
{(1,.86) (2,.90) (3,.93) (4,.88)};
\addplot+[color=black, sharp plot] coordinates
{(1,.88) (2,.94) (3,.94) (4,.92)};
\end{axis}

\end{tikzpicture}
\quad
\begin{tikzpicture}

\begin{axis}[
title=Hit ratio,
ymin=0.1,
ymax=0.7,
symbolic x coords={1,2,3,4},
xticklabels={$\mathcal{G}_1$, $\mathcal{G}_2$, $\mathcal{G}_3$, $\mathcal{G}_4$},
enlargelimits=0.15,
xtick=data,
cycle list name=my black white,
]
\addplot+[color=green, sharp plot] coordinates
{(1,.42) (2,.43) (3,.38) (4,.34)};
\addplot+[color=blue, sharp plot] coordinates
{(1,.67) (2,.57) (3,.60) (4,.50)};
\addplot+[color=red, sharp plot] coordinates
{(1,.67) (2,.61) (3,.59) (4,.56)};
\addplot+[color=black, sharp plot] coordinates
{(1,.68) (2,.63) (3,.58) (4,.55)};
\end{axis}

\end{tikzpicture}
\quad
\begin{tikzpicture}

\begin{axis}[
title=Fairness index,
symbolic x coords={1,2,3,4},
xticklabels={$\mathcal{G}_1$, $\mathcal{G}_2$, $\mathcal{G}_3$, $\mathcal{G}_4$},
enlargelimits=0.15,
ymin=0.6,
xtick=data,
cycle list name=my black white,
]
\addplot+[color=green, sharp plot] coordinates
{(1,1) (2,1) (3,1) (4,1)};
\addplot+[color=blue, sharp plot] coordinates
{(1,.98) (2,.96) (3,.98) (4,.96)};
\addplot+[color=red, sharp plot] coordinates
{(1,.94) (2,.92) (3,.92) (4,.96)};
\addplot+[color=black, sharp plot] coordinates
{(1,.84) (2,.78) (3,.72) (4,.99)};
\end{axis}

\end{tikzpicture}
\\
\ref{named}

\caption{Effect of data sharing changes on four equi-paced tenants on {\sf Sales} workload}

\label{fig:access}
\end{figure*}

\begin{figure}
\centering
\pgfplotsset{width=4.5cm}

\begin{tikzpicture}
\begin{axis}[
legend columns=-1,
legend entries={\mmf,\fastpf,\opt},
legend to name=named,
title=Top 3 views from $g_1$,
ybar,
bar width=5pt,
enlargelimits=0.15,
ymin=0.1,
ymax=0.6,
]
\addplot+[fill=blue, postaction={pattern=north east lines}] plot coordinates
{(1,.4) (2,.57) (3,.33)};
\addplot+[fill=red, postaction={pattern=north west lines}] plot coordinates
{(1,.5) (2,.47) (3,.37)};
\addplot+[fill=black] plot coordinates
{(1,.53) (2,.5) (3,.37)};

\end{axis}
\end{tikzpicture}
\quad%
\begin{tikzpicture}
\begin{axis}[
title=Top 3 views from $g_2$,
ybar,
bar width=5pt,
enlargelimits=0.15,
ymin=0.1,
ymax=0.6,
]
\addplot+[fill=blue, postaction={pattern=north east lines}] plot coordinates
{(1,.37) (2,.37) (3,.23)};
\addplot+[fill=red, postaction={pattern=north west lines}] plot coordinates
{(1,.23) (2,.4) (3,.17)};
\addplot+[fill=black] plot coordinates
{(1,.2) (2,.37) (3,.27)};

\end{axis}
\end{tikzpicture}
\\
\ref{named}

\caption{Fraction of time the popular views in setup $\mathcal{G}_2$ were cached}

\label{fig:cached}
\end{figure}
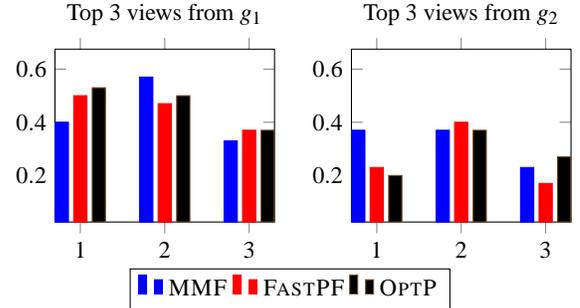

Figure~\ref{fig:access} shows how different algorithms perform in each of these setups. Throuphput goes down with heterogeneity in data access. \static\ performs poorly in all the setups, the performance being between 30\%-40\% worse of the others. Its lower cache utilization and lower hit ratio are further indicators of why \static\ is not the right choice for cache allocation. There is very little to distinguish among the three cache-sharing algorithms. This shows that our fair algorithms can provide a throughput close to the optimal. In terms of fairness, \opt\ algorithm gives the most inconsistent performance. It scores high in the setup with most heterogeneity, but fails when data sharing is involved. \mmf\ and \fastpf, on the other hand, score high in all the setups. 

The performance of \mmf\ interestingly falls alarmingly low in the second setup. This is clearly an outcome of the data sharing pattern wherein three of four tenants largely share the same subset of views. Recollecting the example presented in Table~\ref{table 4}, \mmf\ tries to share the cache (probabilistically) equally between the two sets of tenants effectively producing an allocation off the core. We include a chart showing the duration the most popular views were cached for by \mmf, \fastpf, and \opt. (Figure~\ref{fig:cached}) Top three views in each of $g_1$ and $g_2$ serve 25\%, 13\%, and 8\% of the queries respectively. It can be seen that while \mmf\ caches the topmost view from the distributions roughly equally, \fastpf\ and \opt\ favor the topmost view from $g_1$ more since it is shared by three tenants. \mmf\ tries to compensate the three tenants by caching their second best view more, but this view has a lower utility both due to lower access frequency and smaller size. So the overall performance of \mmf\ suffers in this case.

\subsubsection{Effect of variance in query arrival rates}
\label{sec:arrival}
To replicate the bursty tenants scenarios, we vary query inter-arrival rates of tenants in a two-tenant setup. We create three setups---{\em low}, {\em mid}, and {\em high}---with query inter-arrival rates as listed in Table~\ref{tab:arrival}. The other parameters used in each of the setups are listed in Table~\ref{tab:sarrivalrate}. 

\begin{table}[h!]
\centering
 \begin{tabular}{|c | c | c |} 
 \hline
 Setup & Poisson mean, $\lambda_1$ & Poisson mean, $\lambda_2$ \\ 
 \hline
 {\em low} & 12 & 12 \\
 \hline
 {\em mid} & 18 & 8 \\ 
 \hline
 {\em high} & 24 & 6 \\ 
 \hline
\end{tabular}

\caption{Query inter-arrival rates for different setups}

\label{tab:arrival}
\end{table}

\begin{table}
\centering
 \begin{tabular}{| l | c |} 
 \hline
 Parameter & Value \\ 
 \hline
 Data access distributions & $\{g_1, g_2\}$ \\
 \hline
 Batch size (sec) & 72 \\ 
 \hline
 Number of batches & 30 \\
 \hline
\end{tabular}

\caption{Query arrival rate experiment setup}

\label{tab:sarrivalrate}
\end{table}

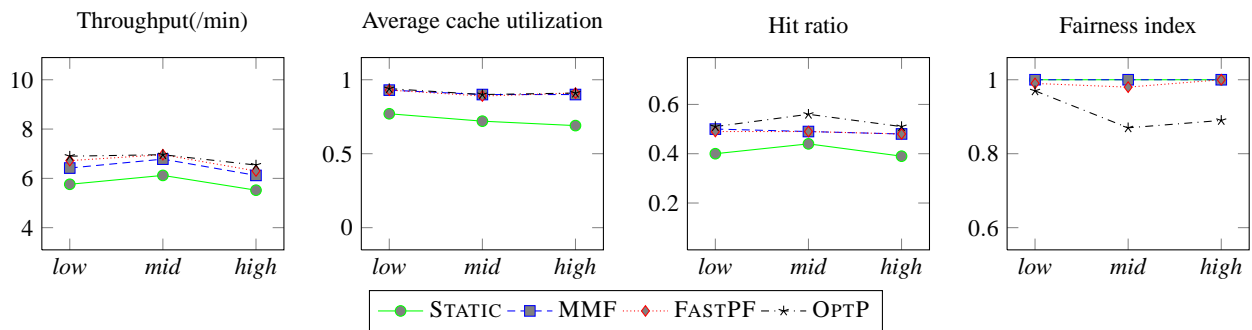
\begin{figure*}
\centering
\pgfplotsset{width=4.8cm}

\begin{tikzpicture}
\begin{axis}[
title={Throughput(/min)},
symbolic x coords={1,2,3},
xticklabels={\em{low}, \em{mid}, \em{high}},
xtick=data,
ymin=4,
ymax=10,
enlargelimits=0.15,
legend columns=-1,
legend entries={\static,\mmf,\fastpf,\opt},
legend to name=named,
cycle list name=my black white,
]
\addplot+[color=green,sharp plot] coordinates
{(1,5.76) (2,6.12) (3,5.52)};
\addplot+[color=blue,sharp plot] coordinates
{(1,6.42) (2,6.78) (3,6.12)};
\addplot+[color=red, sharp plot] coordinates
{(1,6.72) (2,6.96) (3,6.3)};
\addplot+[color=black, sharp plot] coordinates
{(1,6.9) (2,6.96) (3,6.54)};
\end{axis}

\end{tikzpicture}
\quad
\begin{tikzpicture}

\begin{axis}[
symbolic x coords={1,2,3},
xticklabels={\em{low}, \em{mid}, \em{high}},
xtick=data,
title={Average cache utilization},
ymin=0,
ymax=1,
enlargelimits=0.15,
cycle list name=my black white,
]
\addplot+[color=green,sharp plot] coordinates
{(1,0.77) (2,0.72) (3,0.69)};
\addplot+[color=blue,sharp plot] coordinates
{(1,0.93) (2,0.90) (3,0.90)};
\addplot+[color=red, sharp plot] coordinates
{(1,.93) (2,.89) (3,.91)};
\addplot+[color=black, sharp plot] coordinates
{(1,.94) (2,.90) (3,.91)};
\end{axis}

\end{tikzpicture}
\quad
\begin{tikzpicture}

\begin{axis}[
symbolic x coords={1,2,3},
xticklabels={\em{low}, \em{mid}, \em{high}},
xtick=data,
title={Hit ratio},
ymin=0.1,
ymax=0.7,
enlargelimits=0.15,
cycle list name=my black white,
]
\addplot+[color=green,sharp plot] coordinates
{(1,0.40) (2,0.44) (3,0.39)};
\addplot+[color=blue,sharp plot] coordinates
{(1,0.50) (2,0.49) (3,0.48)};
\addplot+[color=red, sharp plot] coordinates
{(1,.49) (2,.49) (3,.48)};
\addplot+[color=black, sharp plot] coordinates
{(1,.51) (2,.56) (3,.51)};
\end{axis}
\end{tikzpicture}
\quad
\begin{tikzpicture}

\begin{axis}[
symbolic x coords={1,2,3},
xticklabels={\em{low}, \em{mid}, \em{high}},
xtick=data,
title={Fairness index},
ymin=0.6,
enlargelimits=0.15,
cycle list name=my black white,
]
\addplot+[color=green,sharp plot] coordinates
{(1,1) (2,1) (3,1)};
\addplot+[color=blue,sharp plot] coordinates
{(1,1) (2,1) (3,1)};
\addplot+[color=red, sharp plot] coordinates
{(1,.99) (2,.98) (3,1)};
\addplot+[color=black, sharp plot] coordinates
{(1,.97) (2,.87) (3,.89)};
\end{axis}

\end{tikzpicture}
\\
\ref{named}

\caption{Effect of variance in query arrival rates}

\label{fig:arrival}
\end{figure*}

Figure~\ref{fig:arrival} shows the impact of variance in query arrival rate on various metrics. The performance of \static\ remains below the other three algorithms as can be seen from the first three graphs. The performance gap, however, is small because the cache is partitioned in only two parts for \static\ each part being large enough to serve 80\% of the queries that could be served off unpartitioned cache. When it comes to the fairness index, all the algorithms except \opt\ get a near-perfect score. \opt\ favors the faster tenant in both {\em mid} and {\em high} setups so much that the slower tenant's performance degrades. Figure~\ref{fig:speedup} shows the speedups for \mmf, \fastpf, and \opt\ relative to \static\ under the setup {\em high}. It can be seen that the first tenant sees a performance degradation with \opt\ empirically proving the fact that \opt\ is not sharing incentive.

\begin{figure}
\centering
\pgfplotsset{width=4.5cm}

\begin{tikzpicture}
\begin{axis}[
legend columns=1,
legend pos=outer north east,
ylabel=mean speedup,
xtick={1,2,3},
xticklabels={\mmf, \fastpf, \opt},
ytick={1},
yticklabels={baseline},
nodes near coords,
every node near coord/.append style={font=\tiny},
enlarge y limits={value=0.2,upper},
ybar,
bar width=5pt,
enlargelimits=0.15,
]
\addplot+[orange, fill=orange, postaction={pattern=north east lines}] plot coordinates
{(1,1.28) (2,1.38) (3,0.67)};
\addplot+[green, fill=green, postaction={pattern=north west lines}] plot coordinates
{(1,1.37) (2,1.36) (3,1.41)};
\draw [black, line width=2pt] (axis cs:\pgfkeysvalueof{/pgfplots/xmin},1)--(axis cs:\pgfkeysvalueof{/pgfplots/xmax},1);

\legend{Tenant 1, Tenant 2}
\end{axis}
\end{tikzpicture}

\caption{Mean speedups provided by different algorithms over \static\ policy for the two tenants in the setup {\em high}}

\label{fig:speedup}
\end{figure}
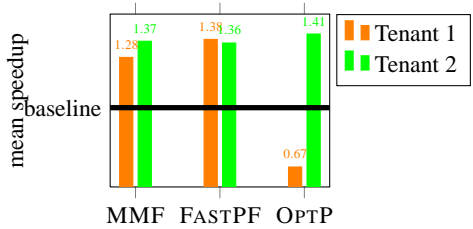

\subsubsection{Effect of number of tenants}
\label{sec:tenants}

\begin{table}[h!]
\centering
 \begin{tabular}{|c | c |} 
 \hline
 Setup & Poisson mean, $\lambda$ \\ 
 \hline
 2 & 10 \\
 \hline
 4 & 20 \\ 
 \hline
 8 & 40 \\ 
 \hline
\end{tabular}

\caption{Query inter-arrival rates for a tenant under different setups}

\label{tab:tenants}
\end{table}

\begin{table}
\centering
 \begin{tabular}{| l | c |} 
 \hline
 Parameter & Value \\ 
 \hline
 Data access distributions & \{$g_1 \forall$ tenant\} \\
 \hline
 Batch size (sec) & 40 \\ 
 \hline
 Number of batches & 30 \\
 \hline
\end{tabular}

\caption{Number of tenants experiment setup}

\label{tab:stenants}
\end{table}

To further stress the utility of optimizing the entire cache as a shared resource, we experiment with increasing number of tenants. Specifically, we consider scenarios with 2, 4, and 8 tenants, all using the same distribution over dataset access. We try to keep the number of queries per batch the same by doubling query inter-arrival rate with doubling of the number of tenants, batch size remaining the same across the setups. Table~\ref{tab:tenants} lists the query inter-arrival rates we used. The other parameters common across the setups are listed in Table~\ref{tab:stenants}.

\pgfplotscreateplotcyclelist{my black white}{%
solid, every mark/.append style={solid, fill=gray}, mark=*\\%
densely dashed, every mark/.append style={solid, fill=gray},mark=square*\\%
densely dotted, every mark/.append style={solid, fill=gray}, mark=diamond*\\%
dashdotted, every mark/.append style={solid, fill=gray},mark=star\\%
dotted, every mark/.append style={solid, fill=gray}, mark=square*\\%
loosely dotted, every mark/.append style={solid, fill=gray}, mark=triangle*\\%
loosely dashed, every mark/.append style={solid, fill=gray},mark=*\\%
dashdotted, every mark/.append style={solid, fill=gray},mark=otimes*\\%
dasdotdotted, every mark/.append style={solid},mark=star\\%
densely dashdotted,every mark/.append style={solid, fill=gray},mark=diamond*\\%
}

\begin{figure*}
\centering
\pgfplotsset{width=4.8cm}

\begin{tikzpicture}
\begin{axis}[
title={Throughput(/min)},
symbolic x coords={2,4,6,8},
xtick=data,
ymin=4,
ymax=10,
enlargelimits=0.15,
legend columns=-1,
legend entries={\static,\mmf,\fastpf,\opt},
legend to name=named,
cycle list name=my black white,
]
\addplot+[color=green, sharp plot] coordinates
{(2,7) (4,6) (8,5.34)};
\addplot+[color=blue, sharp plot] coordinates
{(2,10) (4,9.4) (8,8.34)};
\addplot+[color=red, sharp plot] coordinates
{(2,9.7) (4,9.4) (8,8.22)};
\addplot+[color=black, sharp plot] coordinates
{(2,10.4) (4,10.1) (8,9.18)};
\end{axis}

\end{tikzpicture}
\quad
\begin{tikzpicture}

\begin{axis}[
symbolic x coords={2,4,6,8},
xtick=data,
title={Average cache utilization},
ymin=0,
ymax=1,
enlargelimits=0.15,
cycle list name=my black white,
]
\addplot+[color=green,sharp plot] coordinates
{(2,0.67) (4,0.35) (8,0.07)};
\addplot+[color=blue,sharp plot] coordinates
{(2,0.93) (4,0.87) (8,0.82)};
\addplot+[color=red, sharp plot] coordinates
{(2,.93) (4,.86) (8,.82)};
\addplot+[color=black, sharp plot] coordinates
{(2,.97) (4,.88) (8,.87)};
\end{axis}

\end{tikzpicture}
\quad
\begin{tikzpicture}

\begin{axis}[
symbolic x coords={2,4,6,8},
xtick=data,
title={Hit ratio},
ymin=0.1,
ymax=0.7,
enlargelimits=0.15,
cycle list name=my black white,
]
\addplot+[color=green,sharp plot] coordinates
{(2,.5) (4,.42) (8,.26)};
\addplot+[color=blue,sharp plot] coordinates
{(2,.68) (4,0.67) (8,0.65)};
\addplot+[color=red, sharp plot] coordinates
{(2,.68) (4,.67) (8,.65)};
\addplot+[color=black, sharp plot] coordinates
{(2,.68) (4,.68) (8,.68)};
\end{axis}
\end{tikzpicture}
\quad
\begin{tikzpicture}

\begin{axis}[
symbolic x coords={2,4,6,8},
xtick=data,
title={Fairness index},
ymin=0.6,
enlargelimits=0.15,
cycle list name=my black white,
]
\addplot+[color=green,sharp plot] coordinates
{(2,1) (4,1) (8,1)};
\addplot+[color=blue,sharp plot] coordinates
{(2,0.98) (4,0.98) (8,0.94)};
\addplot+[color=red, sharp plot] coordinates
{(2,1) (4,.94) (8,.91)};
\addplot+[color=black, sharp plot] coordinates
{(2,1) (4,.84) (8,.78)};
\end{axis}

\end{tikzpicture}
\\
\ref{named}

\caption{Effect of changing number of tenants}

\label{fig:tenants}
\end{figure*}
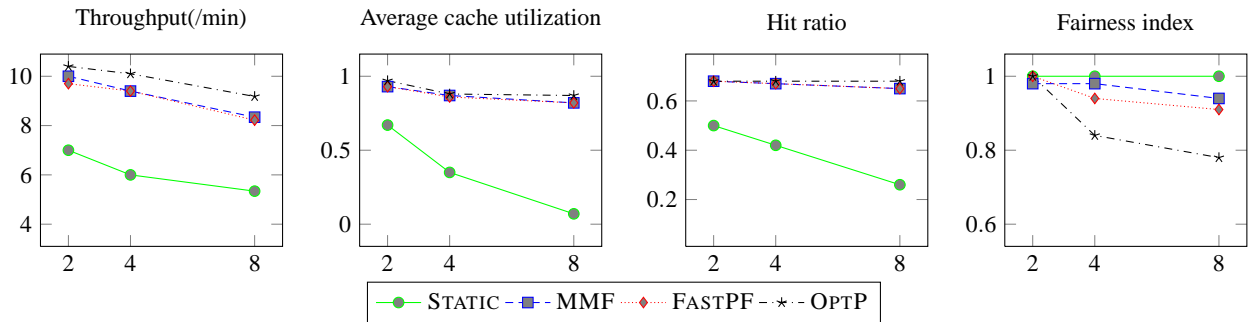

Figure~\ref{fig:tenants} shows behavior of the algorithms under these scenarios. The gap in throughput between \static\ and the other algorithms is large (35\%-45\%). As the number of tenants goes up, the average cache utilization of \static\ drops sharply, whereas the average cache utilization of the other algorithms remain largely stable. This can be attributed to the static partitioning of cache in \static. The hit ratio shows a similar pattern again showing why \static\ is not the best choice. In terms of fairness index, \opt\ finds it increasingly harder to provide a fair solution. With an increase in the number of tenants, the number of queries per tenant per batch goes down which makes the locally optimal choices of \opt\ more unlikely to provide equal speedups. In contrast, \mmf\ and \fastpf, with their randomized choices, score over 0.9 in all the scenarios exhibiting their superiority.

\subsection{Discussion and Future Work}
\label{sec:evaldiscuss}
Our evaluation on practical setups  brings up some interesting insights that opens up multiple possibilities for the future. We discuss some of the challenges and the directions here.

Our experiments show that across all setups, \fastpf\ and \mmf\ provide far better trade-offs in throughput and fairness compared to \static\ and \opt.  We note that in comparing max-min fair and proportional fair implementations, there is no clear winner. We believe this is a second order difference that a more precise cost model and implementations of the exact algorithms (for instance, the algorithm in Section~\ref{sec:pfalg} for proportional fairness) will bring out. However, even given similar empirical results, PF has the advantage of the core property as a succinct and easy to explain notion of fairness.

We next note that the {\bf running time} of our algorithms is polynomial in number of tenants. In most typical industry setups, the ones we evaluated, there is only a handful number of tenants. Therefore, we expect our algorithms to be fast even in the wild. Just to quantify the query wait times, we observed them to be of the order of tens of milliseconds in most cases.

\paragraph{Convergence Properties}
\pgfplotscreateplotcyclelist{my black white}{%
densely dashed, every mark/.append style={solid, fill=gray},mark=square*\\%
densely dotted, every mark/.append style={solid, fill=gray}, mark=diamond*\\%
solid, every mark/.append style={solid, fill=gray}, mark=triangle*\\%
dashdotted, every mark/.append style={solid, fill=gray},mark=star\\%
dotted, every mark/.append style={solid, fill=gray}, mark=square*\\%
loosely dotted, every mark/.append style={solid, fill=gray}, mark=triangle*\\%
loosely dashed, every mark/.append style={solid, fill=gray},mark=*\\%
dashdotted, every mark/.append style={solid, fill=gray},mark=otimes*\\%
dasdotdotted, every mark/.append style={solid},mark=star\\%
densely dashdotted,every mark/.append style={solid, fill=gray},mark=diamond*\\%
}
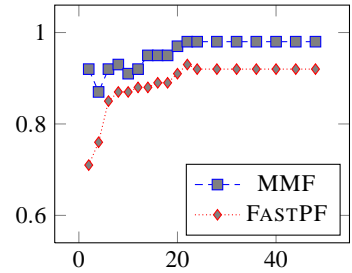
\begin{figure}
\pgfplotsset{width=5.5cm}
\centering
\begin{tikzpicture}
\begin{axis}[
ymin=0.6, ymax=1,
enlargelimits=0.15,
legend pos=south east,
cycle list name=my black white,
]
\addplot+[color=blue,sharp plot] coordinates
{(2,0.92) (4,0.87) (6,0.92) (8,0.93) (10,0.91) (12,0.92) (14,0.95) (16,0.95) (18,0.95) (20,0.97) (22,0.98) (24,0.98) (28,0.98) (32,0.98) (36,0.98) (40,0.98) (44,0.98) (48,0.98)};
\addplot+[color=red,sharp plot] coordinates
{(2,0.71) (4,0.76) (6,0.85) (8,0.87) (10,0.87) (12,0.88) (14,0.88) (16,0.89) (18,0.89) (20,0.91) (22,0.93) (24,0.92) (28,0.92) (32,0.92) (36,0.92) (40,0.92) (44,0.92) (48,0.92)};
\addlegendentry{\mmf}
\addlegendentry{\fastpf}
\end{axis}
\end{tikzpicture}

 \caption{Fairness index as a function of number of batches}

 \label{fig:conv}
\end{figure}

As our algorithms are randomized in nature, it is important to study how long they take to converge to solutions that yield fairness across time. After running several workloads, we find that the number of batches to achieve convergence is very small, of the order of 15-25. In Figure~\ref{fig:conv}, we present results of a four tenant workload with 50 batches, optimized once using \mmf\ and once using \fastpf. The fairness index was computed after every 2-4 batches. It can be seen that both algorithms converge to their respective optimal values at around 20 batches. As a future work, we plan to systematically study which parameters define rate of convergence of the algorithms.

\paragraph{Batch Size and Cache State}
Our batched processing architecture introduces additional optimization choices. Primarily, there are two ways of tuning a view selection algorithm: \\
1. Controlling batch size, and \\
2. Managing state of cache across batches.

The first option needs no elaboration. The second option is whether the cache is treated as {\em stateful} or as {\em stateless} when optimizing a batch. In the former case, the estimated benefit of views that are already in cache is boosted by a factor $\gamma > 1$. This influences the next cache allocation, and makes it more likely for these views to stay in the cache. The latter case ignores the state of the cache when considering the next batch. All the results presented so far have used stateless cache.

\pgfplotscreateplotcyclelist{my black white}{%
densely dashed, every mark/.append style={solid, fill=gray},mark=square*\\%
solid, every mark/.append style={solid, fill=gray}, mark=triangle*\\%
densely dotted, every mark/.append style={solid, fill=gray}, mark=diamond*\\%
dashdotted, every mark/.append style={solid, fill=gray},mark=otimes\\%
dotted, every mark/.append style={solid, fill=gray}, mark=square*\\%
loosely dotted, every mark/.append style={solid, fill=gray}, mark=triangle*\\%
loosely dashed, every mark/.append style={solid, fill=gray},mark=*\\%
dashdotted, every mark/.append style={solid, fill=gray},mark=otimes*\\%
dasdotdotted, every mark/.append style={solid},mark=star\\%
densely dashdotted,every mark/.append style={solid, fill=gray},mark=diamond*\\%
}

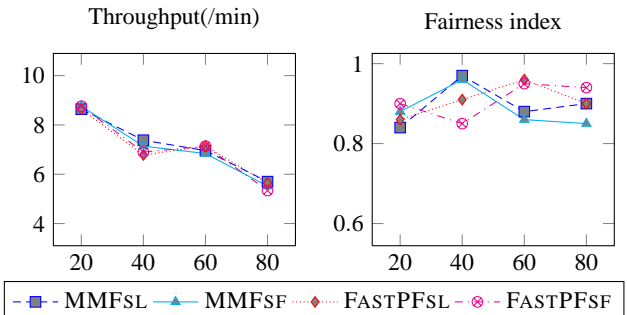
\begin{figure}
\centering
\pgfplotsset{width=4.8cm}

\begin{tikzpicture}

\begin{axis}[
symbolic x coords={20,40,60,80},
xtick=data,
ymin=4,
ymax=10,
title=Throughput(/min),
enlargelimits=0.15,
legend columns=-1,
legend entries={\mmfsl,\mmfsf,\pfsl,\pfsf},
legend to name=named,
cycle list name=my black white,
]
\addplot+[color=blue, sharp plot] coordinates
{(20,8.64) (40,7.38) (60,6.96) (80,5.7)};
\addplot+[color=cyan, sharp plot] coordinates
{(20,8.76) (40,7.14) (60,6.84) (80,5.52)};
\addplot+[color=red, sharp plot] coordinates
{(20,8.64) (40,6.78) (60,7.14) (80,5.64)};
\addplot+[color=magenta, sharp plot] coordinates
{(20,8.76) (40,6.9) (60,7.14) (80,5.34)};
\end{axis}

\end{tikzpicture}
\quad
\begin{tikzpicture}

\begin{axis}[
symbolic x coords={20,40,60,80},
xtick=data,
title=Fairness index,
ymin=0.6,
enlargelimits=0.15,
cycle list name=my black white,
]
\addplot+[color=blue, sharp plot] coordinates
{(20,0.84) (40,0.97) (60,0.88) (80,0.9)};
\addplot+[color=cyan, sharp plot] coordinates
{(20,0.88) (40,0.96) (60,0.86) (80,0.85)};
\addplot+[color=red, sharp plot] coordinates
{(20,0.86) (40,0.91) (60,0.96) (80,0.9)};
\addplot+[color=magenta, sharp plot] coordinates
{(20,0.9) (40,0.85) (60,0.95) (80,0.94)};
\end{axis}

\end{tikzpicture}
\\
\ref{named}

\caption{Effect of batch size on four equi-paced tenants setup}

\label{fig:batchsize}
\end{figure}

We empirically compared how the algorithms react to these parameters. Figure~\ref{fig:batchsize} shows effect of change in batch size on two versions each of \mmf\ and \fastpf: one treating the cache as stateless (\mmfsl\ and \pfsl), and the other treating it as stateful (\mmfsf\ and \pfsf), with $\gamma = 2$. It can be seen that both versions provide similar throughput in all the cases. It can be observed that the stateful algorithms score higher on fairness for the smallest batch size but there is no clear pattern seen when the batch size is larger. It makes sense since the lower batch sizes do not give enough choices for fair configurations of cache and maintaining the state results in an artificial increase of the batch size thereby providing better configurations. As a future work, we plan to explore these trade-offs on a larger scale to devise better guidelines on parameter tuning.

\paragraph{Engineering issues}
We now highlight some challenges in scaling up our experiments to industry scale. These challenges are tied to engineering issues in current implementations of systems such as Spark, and will get ironed out over time. Most common multi-tenant Spark setups use a separate Spark {\em context} for each tenant, effectively partitioning cache. In fact, most current multi-tenant data warehouse systems recommend splitting memory across queues. This is in part due to multi-thread management challenges that result in unpredictable behavior such as premature eviction of cached data blocks. Another engineering issue, specific to Spark, is the inordinately long delays in garbage collection when cluster scales up. 
We should be able to see a much better impact of \robus\ optimization once these practical issues get resolved. 

\paragraph{Code Base}
The code base of \robus\ has been open-sourced~\cite{robus-github} and our entire experimental setup can be replicated following a simple set of instructions provided with the code.

%% file: related.tex
\section{Related Work}
\label{sec:related}

\paragraph{Physical design tuning and Multi-query optimization} Classical view materialization algorithms in databases~\cite{gupta, yang, automated, Goldstein, mistry} treat entire workload as a set and optimize towards one or more of the flow time, space budget, and view maintenance costs. Online physical design tuning approaches~\cite{tuning, dynamat, restore, miso}, on the other hand, adapt to changes in query workload by modifying physical design. None of the afore-mentioned approaches support multi-tenant workloads and therefore cannot be used in selecting views for caching. However, some of the techniques used, in particular candidate view enumeration, view matching, and query rewrite, can be applied in \robus\ framework.

Batched optimization of queries was proposed in~\cite{sellis} and is used in many work sharing approaches~\cite{cooperative,sharedscan,mrshare}. \robus\ employs batched query optimization likewise, but crucially also ensures that each tenant gets their fair share of benefit.

\paragraph{Fairness theory}
The proportional fairness algorithm is widely studied in Economics~\cite{bargaining,eisenberg,budish} as well as in scheduling theory~\cite{drf,beyonddrf,KMT,KellyMW09,ImKM14,Stolyar,AndrewsPF}. In the context of resource partitioning problems (or exchange economies)~\cite{continuum,exchange}, it is well-known that a convex program, called the Eisenberg-Gale convex program~\cite{eisenberg} computes prices that implement a Walrasian equilibrium (or market clearing solution). Our shared resource allocation problem is different from allocation problems where resources need to be partitioned, and it is not clear how to specify prices for resources (or views) in our setting. Nevertheless, we show that there is an exponential size convex program using configurations as variables, whose solution implements proportional fairness in a  randomized sense. 

In scheduling theory, the focus is on analyzing delay properties~\cite{KMT,KellyMW09,ImKM14} assuming jobs have durations. Our focus is instead on utility maximization, which has also been considered in the context of wireless scheduling in~\cite{Stolyar,AndrewsPF}. The latter work focuses on {\em long-term} fairness for partitioned resources, where utility of a tenant is defined as sum of discounted utilities across time. The resulting algorithms, though simple, only provide guarantees assuming job arrivals are ergodic, and if tenants exist forever. They do not provide per-epoch guarantees. In contrast, we focus on obtaining per-epoch fairness in a randomized sense without ergodic assumptions, and on defining the right fairness concepts when resources are shared. We finally note that~\cite{noagentleftbehind} presents dynamic schemes for achieving envy-freeness across time; however, these techniques are specific to resource partitioning problems and to not directly apply to our shared resource setting.

\paragraph{Multi-tenant architectures}
Traditionally, the notion of multi-tenancy in databases deals with sharing database system resources, viz. hardware, process, schema, among users~\cite{multi-tenant, multi-tenant2, sqlvm}. Each tenant only accesses data owned by them. Emerging multi-tenant big data architectures, on the other hand, allow for entire cluster data to be shared among tenants. This sharing of data is critical in our work as it allows the cache to be used much more efficiently.

A critical component of modern multi-tenant architectures, such as Apache Hadoop, Apache Spark, Cloudera Impala, is a fair scheduler/ resource allocator~\cite{hadoop-fair, drf, quincy}. The resource pool considered by these schedulers do not differentiate the cache resource from the heap resource and as a result divides the cache among tenants. As seen in our work, partitioned cache setups severely reduce optimization opportunities. Some recent approaches treat cache as an independent resource when running multiple jobs. PACMan~\cite{pacman} exploits multi-wave execution workflow of Hadoop jobs to make caching decisions at the granularity of parallely running tasks of a job. In another work, LRU policies for Buffer pool memory are extended to meet SLA guarantees of multiple tenants~\cite{bufferpool}. However, none of the approaches exploit the opportunities presented by multi-shared nature of cache resource. In another advancement, distributed analytics systems are supporting a distributed cache store shared by multiple tenants~\cite{tachyon, ddm}. \robus\ optimizer will be a natural fit for such systems.

%% file: conclusion.tex
\section{Conclusion} 
\label{sec:conclusion}
 Emerging Big data multi-tenant analytics systems complement an abundant disk-based storage with a smaller, but much faster, cache in order to optimize workloads by materializing views in the cache. The cache is a shared resource, i.e., cached data can be accessed by all tenants. In this paper, we presented \robus, a cache management platform for achieving both a fair allocation of cache and a near-optimal performance in such architectures. We defined notions of fairness for the shared settings using randomization in small batches as a key tool. We presented a fairness model that incorporates Pareto-efficiency and sharing incentive, and also achieves envy-freeness via the notion of core from cooperative game theory. We showed a proportionally fair mechanism to satisfy the core property in expectation. Further, we developed efficient algorithms for two fair mechanisms and implemented them in a \robus\ prototype built on a Spark cluster. Our experiments on various practical setups show that it is possible to achieve near-optimal fairness, while simultaneously preserving near-optimal performance speedups using the algorithms we developed. 
 
Our framework is quite general and applies to any setting where resource allocations are shared across agents. As future work, we plan to explore other applications of this framework.

%% file: appendix.tex
\section{Experiment results}

\subsection{Results of experiments on effect of data sharing on mixed workload}

\begin{table}[h!]
\centering
\resizebox{\columnwidth}{!}{
\begin{tabular}{| l | c | c | c | c | } 
 \hline
 Metric & \static & \mmf & \fastpf & \opt \\
 \hline
 Throughput(/min) & 7.80 & 19.2 & 19.2 & 19.2 \\
 \hline
 Avg cache util. & 0.00 & 0.83 & 0.83 & 0.83 \\
 \hline
 Hit ratio & 0.00 & 1.00 & 1.00 & 1.00 \\
 \hline
 Fairness index & 1.00 & 0.71 & 0.71 & 0.71 \\
 \hline
\end{tabular}
}
\caption{Performance of algorithms on setup $\mathcal{G}_1$}

\label{tab:access1}
\end{table}

\begin{table}[h!]
\centering
\resizebox{\columnwidth}{!}{
\begin{tabular}{| l | c | c | c | c | } 
 \hline
 Metric & \static & \mmf & \fastpf & \opt \\
 \hline
 Throughput(/min) & 7.20 & 9.00 & 10.2 & 16.2 \\
 \hline
 Avg cache util. & 0.08 & 0.81 & 0.87 & 0.92 \\
 \hline
 Hit ratio & 0.08 & 0.54 & 0.68 & 0.83 \\
 \hline
 Fairness index & 1.00 & 0.83 & 0.79 & 0.75 \\
 \hline
\end{tabular}
}
\caption{Performance of algorithms on setup $\mathcal{G}_2$}

\label{tab:access2}
\end{table}

\begin{table}[h!]
\centering
\resizebox{\columnwidth}{!}{
\begin{tabular}{| l | c | c | c | c | } 
 \hline
 Metric & \static & \mmf & \fastpf & \opt \\
 \hline
 Throughput(/min) & 7.20 & 7.50 & 7.80 & 9.60 \\
 \hline
 Avg cache util. & 0.16 & 0.96 & 0.98 & 1.00 \\
 \hline
 Hit ratio & 0.19 & 0.53 & 0.55 & 0.67 \\
 \hline
 Fairness index & 1.00 & 0.77 & 0.66 & 0.50 \\
 \hline
\end{tabular}
}
\caption{Performance of algorithms on setup $\mathcal{G}_3$}

\label{tab:access3}
\end{table}

\begin{table}[h!]
\centering
\resizebox{\columnwidth}{!}{
\begin{tabular}{| l | c | c | c | c | } 
 \hline
 Metric & \static & \mmf & \fastpf & \opt \\
 \hline
 Throughput(/min) & 5.40 & 5.40 & 5.40 & 4.80 \\
 \hline
 Avg cache util. & 0.24 & 0.91 & 0.93 & 0.96 \\
 \hline
 Hit ratio & 0.26 & 0.43 & 0.47 & 0.46 \\
 \hline
 Fairness index & 1.00 & 0.81 & 0.80 & 0.38 \\
 \hline
\end{tabular}
}
\caption{Performance of algorithms on setup $\mathcal{G}_4$}

\label{tab:access4}
\end{table}

\FloatBarrier

\subsection{Results of experiments on effect of data sharing on Sales workload}

\begin{table}[h!]
\centering
\resizebox{\columnwidth}{!}{
\begin{tabular}{| l | c | c | c | c | } 
 \hline
 Metric & \static & \mmf & \fastpf & \opt \\
 \hline
 Throughput(/min) & 6.00 & 9.42 & 9.42 & 10.08 \\
 \hline
 Avg cache util. & 0.34 & 0.87 & 0.86 & 0.88 \\
 \hline
 Hit ratio & 0.42 & 0.67 & 0.67 & 0.68 \\
 \hline
 Fairness index & 1.00 & 0.98 & 0.94 & 0.84 \\
 \hline
\end{tabular}
}
\caption{Performance of algorithms on setup $\mathcal{G}_1$}

\label{tab:access1}
\end{table}

\begin{table}[h!]
\centering
\resizebox{\columnwidth}{!}{
\begin{tabular}{| l | c | c | c | c | } 
 \hline
 Metric & \static & \mmf & \fastpf & \opt \\
 \hline
 Throughput(/min) & 5.70 & 7.20 & 7.44 & 8.24 \\
 \hline
 Avg cache util. & 0.34 & 0.93 & 0.90 & 0.94 \\
 \hline
 Hit ratio & 0.43 & 0.57 & 0.61 & 0.63 \\
 \hline
 Fairness index & 1.00 & 0.96 & 0.92 & 0.78 \\
 \hline
\end{tabular}
}
\caption{Performance of algorithms on setup $\mathcal{G}_2$}

\label{tab:access2}
\end{table}

\begin{table}[h!]
\centering
\resizebox{\columnwidth}{!}{
\begin{tabular}{| l | c | c | c | c | } 
 \hline
 Metric & \static & \mmf & \fastpf & \opt \\
 \hline
 Throughput(/min) & 5.34 & 7.44 & 7.38 & 7.92 \\
 \hline
 Avg cache util. & 0.30 & 0.93 & 0.93 & 0.94 \\
 \hline
 Hit ratio & 0.38 & 0.60 & 0.59 & 0.58 \\
 \hline
 Fairness index & 1.00 & 0.98 & 0.92 & 0.72 \\
 \hline
\end{tabular}
}
\caption{Performance of algorithms on setup $\mathcal{G}_3$}

\label{tab:access3}
\end{table}

\begin{table}[h!]
\centering
\resizebox{\columnwidth}{!}{
\begin{tabular}{| l | c | c | c | c | } 
 \hline
 Metric & \static & \mmf & \fastpf & \opt \\
 \hline
 Throughput(/min) & 4.20 & 5.64 & 5.76 & 6.00 \\
 \hline
 Avg cache util. & 0.28 & 0.89 & 0.88 & 0.92 \\
 \hline
 Hit ratio & 0.34 & 0.50 & 0.56 & 0.55 \\
 \hline
 Fairness index & 1.00 & 0.96 & 0.96 & 0.99 \\
 \hline
\end{tabular}
}
\caption{Performance of algorithms on setup $\mathcal{G}_4$}

\label{tab:access4}
\end{table}

\FloatBarrier

\subsection{Results of experiments on effect of query arrival rate}

\begin{table}[h!]
\centering
\resizebox{\columnwidth}{!}{
\begin{tabular}{| l | c | c | c | c | } 
 \hline
 Metric & \static & \mmf & \fastpf & \opt \\
 \hline
 Throughput(/min) & 5.76 & 6.42 & 6.72 & 6.90 \\
 \hline
 Avg cache util. & 0.77 & 0.93 & 0.93 & 0.94 \\
 \hline
 Hit ratio & 0.40 & 0.50 & 0.49 & 0.51 \\
 \hline
 Fairness index & 1.00 & 1.00 & 0.99 & 0.97 \\
 \hline
\end{tabular}
}
\caption{Performance of algorithms on setup {\em low}}

\label{tab:arrival1}
\end{table}

\begin{table}[h!]
\centering
\resizebox{\columnwidth}{!}{
\begin{tabular}{| l | c | c | c | c | } 
 \hline
 Metric & \static & \mmf & \fastpf & \opt \\
 \hline
 Throughput(/min) & 6.12 & 6.78 & 6.96 & 6.96 \\
 \hline
 Avg cache util. & 0.72 & 0.90 & 0.89 & 0.90 \\
 \hline
 Hit ratio & 0.44 & 0.49 & 0.49 & 0.56 \\
 \hline
 Fairness index & 1.00 & 1.00 & 0.98 & 0.87 \\
 \hline
\end{tabular}
}
\caption{Performance of algorithms on setup {\em mid}}

\label{tab:arrival2}
\end{table}

\begin{table}[h!]
\centering
\resizebox{\columnwidth}{!}{
\begin{tabular}{| l | c | c | c | c | } 
 \hline
 Metric & \static & \mmf & \fastpf & \opt \\
 \hline
 Throughput(/min) & 5.52 & 6.12 & 6.30 & 6.54 \\
 \hline
 Avg cache util. & 0.69 & 0.90 & 0.91 & 0.91 \\
 \hline
 Hit ratio & 0.39 & 0.48 & 0.48 & 0.51 \\
 \hline
 Fairness index & 1.00 & 1.00 & 1.00 & 0.89 \\
 \hline
\end{tabular}
}
\caption{Performance of algorithms on setup {\em high}}

\label{tab:arrival3}
\end{table}

\FloatBarrier

\subsection{Results of experiments on effect of number of tenants}

\begin{table}[h!]
\centering
\resizebox{\columnwidth}{!}{
\begin{tabular}{| l | c | c | c | c | } 
 \hline
 Metric & \static & \mmf & \fastpf & \opt \\
 \hline
 Throughput(/min) & 7.00 & 10.00 & 9.70 & 10.40 \\
 \hline
 Avg cache util. & 0.67 & 0.93 & 0.93 & 0.97 \\
 \hline
 Hit ratio & 0.50 & 0.68 & 0.68 & 0.68 \\
 \hline
 Fairness index & 1.00 & 0.98 & 1.00 & 1.00 \\
 \hline
\end{tabular}
}
\caption{Performance of algorithms on setup with 2 tenants}

\label{tab:tenants1}
\end{table}

\begin{table}[H]
\centering
\resizebox{\columnwidth}{!}{
\begin{tabular}{| l | c | c | c | c | } 
 \hline
 Metric & \static & \mmf & \fastpf & \opt \\
 \hline
 Throughput(/min) & 6.00 & 9.40 & 9.40 & 10.10 \\
 \hline
 Avg cache util. & 0.34 & 0.87 & 0.86 & 0.88 \\
 \hline
 Hit ratio & 0.42 & 0.67 & 0.67 & 0.68 \\
 \hline
 Fairness index & 1.00 & 0.98 & 0.94 & 0.84 \\
 \hline
\end{tabular}
}
\caption{Performance of algorithms on setup with 4 tenants}

\label{tab:tenants2}
\end{table}

\begin{table}[H]
\centering
\resizebox{\columnwidth}{!}{
\begin{tabular}{| l | c | c | c | c | } 
 \hline
 Metric & \static & \mmf & \fastpf & \opt \\
 \hline
 Throughput(/min) & 5.34 & 8.34 & 8.22 & 9.18 \\
 \hline
 Avg cache util. & 0.07 & 0.82 & 0.82 & 0.87 \\
 \hline
 Hit ratio & 0.26 & 0.65 & 0.65 & 0.68 \\
 \hline
 Fairness index & 1.00 & 0.94 & 0.91 & 0.78 \\
 \hline
\end{tabular}
}
\caption{Performance of algorithms on setup with 8 tenants}

\label{tab:tenants3}
\end{table}

\FloatBarrier